\documentclass[12pt]{article}
\usepackage[letterpaper, nohead,dvips]{geometry}
\usepackage{amsmath}
\usepackage{amsfonts}
\usepackage{amssymb}
\usepackage{enumerate}
\usepackage{amsthm}
\usepackage{setspace}
\usepackage[all,cmtip]{xy}
\usepackage{graphicx}

\renewcommand{\geq}{\geqslant}
\renewcommand{\vec}{\boldsymbol}
\newcommand{\rmd}{\mathrm{d}}

\newcommand{\Eset}{\mathbb{E}}
\newcommand{\Hset}{\mathbb{H}}
\newcommand{\Mset}{\mathbb{M}}
\newcommand{\Rset}{\mathbb{R}}
\newcommand{\Sset}{\mathbb{S}}

\newcommand{\cK}{\mathcal{K}}
\newcommand{\cM}{\mathcal{M}}

\newcommand{\cS}{\mathcal{S}}

\newcommand{\SE}{\mathrm{SE}}
\newcommand{\SO}{\mathrm{SO}}

\def\mathbi#1{\textbf{\em #1}}
\theoremstyle{plain}
\newtheorem{theorem}{Theorem}[section]
\newtheorem{proposition}[theorem]{Proposition}
\newtheorem{lemma}[theorem]{Lemma}
\theoremstyle{definition}
\newtheorem{definition}[theorem]{Definition}
\newtheorem{remark}[theorem]{Remark}

\numberwithin{equation}{section}
\numberwithin{figure}{section}

\begin{document}

\begin{center}\Large Equivalence problem for the orthogonal webs on the sphere\end{center}

\begin{center}\large Caroline M. Cochran\footnote{Department of
  Mathematics and Statistics, Dalhousie University, Halifax, Nova Scotia,
  B3H~3J5, Canada, email:\ adlamc@mathstat.dal.ca},
  Raymond G.\ McLenaghan\footnote{Department of Applied Mathematics, University
  of Waterloo, Waterloo, Ontario, N2L~3G1, Canada,
  email:\ rgmclena@uwaterloo.ca} and Roman G.\ Smirnov\footnote{Department of
  Mathematics and Statistics, Dalhousie University, Halifax, Nova Scotia,
  B3H~3J5, Canada, email:\ smirnov@mathstat.dal.ca}\end{center}

\begin{center}  \end{center}

\begin{quote}
  {\small \textbf{Abstract.} }

  We solve the equivalence problem for the orthogonally separable webs on the
three-sphere  under the action of  the isometry group.
This continues a classical project initiated by Olevsky in which he
solved the corresponding canonical forms problem.  The solution to the
equivalence problem together with the results by Olevsky forms a
complete solution to the problem of orthogonal separation of variables to
the Hamilton-Jacobi equation defined on the three-sphere via orthogonal
separation of variables. It  is based on invariant properties of the characteristic
Killing two-tensors in addition to properties of the corresponding algebraic curvature tensor and the associated Ricci tensor.  The result is
illustrated by a non-trivial application to a natural Hamiltonian defined
on the three-sphere.
\end{quote}


\section{Introduction} \label{sec:intro}
 This paper presents a complete solution to the {\em equivalence problem} for {\em orthogonally separable webs} (OSWs) generated by the {\em characteristic Killing tensors}  (CKTs) defined on the  three-dimensional sphere $\Sset^3$.  In addition, we lay the ground work for solving the analogous problem for three-dimensional hyperbolic space $\Hset^3$ as well as the spaces of constant non-zero curvature of higher dimensions. This work  continues and generalizes our studies of the
OSWs defined in flat (pseudo-)Riemannian spaces (see, for instance, \cite{AMS07, HMS05, HMS09} and the relevant references therein).

Reportedly, the study of
OSWs generated by Killing two-tensors was conceived  in various classical articles published throughout the 19th century. A list of such works should include, among many others, Liouville \cite{Li1846},  Neumann \cite{N1856}, Bertrand \cite{Be1857}, Morera \cite{Mo1881}, and St\"{a}ckel \cite{St1891}. The research in the area continued throughout the last century to make this study more systematic and develop new results. Of note was the celebrated 1934 paper by Eisenhart \cite{E34}. In this paper, the author completely solved the {\em canonical forms problem} for the
OSWs generated by CKTs defined in $\Eset^3$ by showing  that there were exactly
eleven inequivalent (in some sense) such
OSWs and represented each of them with the corresponding canonical form. Furthermore, he gave a general criterion linking orthogonal separability with geometric and algebraic properties of Killing two-tensors.  In addition to $\Eset^3$, he also partially solved the canonical forms problem for $\Sset^3$.  This problem, as well as the case of hyperbolic space $\Hset^3$, was later completed by Olevsky in 1950. In his paper \cite{O50}, Olevsky determined the number of orbits corresponding to the orthogonal separable coordinates on $\Sset^3$ and $\Hset^3$, listing the metrics in these coordinates and the transformations to separable coordinates in each case. Kalnins {\em et al} \cite{KMW76} completed Olevsky's work on $\Sset^3$ by also listing the CKTs (using the language of differential operators) in their respective canonical forms representing the
OSWs. In recent years, the research in the area has been  successfully continued yielding many important results in the area in connection with the study of integrable and superintegrable classical and quantum Hamiltonian systems (see, for example, \cite{BKM76, HMS05, HMS09,  Ka86, RW05, WF65} and the relevant references therein).

It must be noted, however, that many of the aforementioned investigations concerned the canonical forms problems for the CKTs in question, while the applications arising in the field of classical and quantum Hamiltonian systems often involve the properties of the CKTs in their general form as far as the action of the corresponding isometry groups are concerned (see \cite{HMS05}, for example). From this perspective, the main goal of this paper is to extend the results by Olevsky and others and solve the equivalence problem for the CKTs defined on $\Sset^3$.

The paper is organized as follows. Section 2 contains a brief review of the required mathematical tools defined in the framework of the invariant theory of Killing tensors (ITKTs) that will be used in solving the main problem. In Section 3 we determine the form of the general Killing tensor on $\Sset^n$ by considering it as an imbedded hypersurface in $\Eset^{n+1}$.  We also obtain algebraic conditions for the orthogonal integrability of the eigendirections of the CKTs. In Section 4 we formulate and solve the equivalence problems for CKTs defined on  $\Sset^3$. Section 5 is devoted to applications to problems of classical mechanics that subsume the results of the previous sections. Section 6 contains the conclusion.

\section{Invariant theory of Killing tensors}
\label{section1}
In what follows we formulate and solve the equivalence problem for the OSWs generated by CKTs defined on $\Sset^3$ within the framework of the invariant theory of Killing tensors (ITKTs), which is a natural extension of the classical invariant theory of homogeneous polynomials (see \cite{H08} for more details and relevant references). Thus, the study of OSWs is based on algebraic and geometric properties of Killing tensors two-tensors in which invariant theory comes into play as the natural link between algebra and geometry. Moreover, recently it has been shown {\em explicitly} (see \cite{HMS09} for more details and references) that in the case of Killing two-tensors generating  OSWs the study can be  naturally cast into  the general setting of Cartan's geometry \cite{EC37, EC01, PG74}.

Let $(\cM, \vec{g})$ be an $n$-dimensional (pseudo-)Riemannian manifold of
constant curvature.
\begin{definition} \label{KT}

A (contravariant) {\em Killing tensor of valence $p$}  defined in
  $(\cM, \vec{g})$ is a symmetric $(p,0)$ tensor field satisfying the Killing
  tensor equation
  \begin{equation}
    [\vec{K}, \vec{g}] = 0, \label{KTeqn}
  \end{equation}
  where $[$ , $]$ denotes the Schouten bracket \cite{Sc40}.
  When $p=1$, $\vec{K}$ is said to be a Killing vector field (infinitesimal
  isometry) and (\ref{KTeqn}) reduces to
  \begin{equation}
    \mathcal{L}_{\vec{K}} \vec{g} = 0, \label{eq:HJ:KVeqn}
  \end{equation}
  where $\mathcal{L}$ denotes the the Lie derivative operator.

\end{definition}
Since the Schouten bracket $[$ , $]$ is $\mathbb{R}$-bilinear, the set of solutions
to the system of overdetermined partial differential equations (PDEs) given by
(\ref{KTeqn}) forms a vector space over $\mathbb{R}$. Furthermore, since
$(\cM, \vec{g})$ is of constant curvature, the dimension of such a vector
space is maximal (see the relevant references in \cite{HMS05} for more
details). In what follows, we shall use the notation $\cK^{p}(\cM)$ to denote
the vector space of valence $p$ Killing tensor fields defined on $\cM$.
\begin{remark} \label{C_KT}

Note that the equation (\ref{KTeqn}) can be equivalently rewritten in the more familiar (covariant) form as follows:

\begin{equation} \label{CKTeqn}
K_{(i_1,\ldots,i_p;i_{p+1})} = 0,
\end{equation}
where $K_{i_1,\ldots,i_p}$ denotes the covariant components of the Killing tensor $\vec{K}$, $;$ the covariant derivative with respect to the Levi-Civita connection defined by $\vec{g}$, and $(\ldots)$ symmetrization over the enclosed indices
\end{remark}
As is well-known the Killing tensors defined on spaces of constant curvature are sums of symmetrized tensor products of Killing vectors forming finite-dimensional  vector spaces. Denote by $d$ the dimension of the vector space  $\cK^{p}(\cM)$ and recall that $d$ is given by
\begin{equation}
  d = \dim \mathcal{K}^{p}(\cM) = \frac{1}{n} \binom{n+p}{p+1} \binom{n+p-1}{p},
    \quad p \geq 1. \label{dimKT}
\end{equation}
Therefore the general element of $\cK^{p}(\cM)$ may be represented by $d$ arbitrary parameters $a^{1}, \ldots, a^{d}$
with respect to a given basis.  Let $G$ denote the isometry group of $(\cM,\vec{g})$.  
Our next observation is that the group $G$ acting
on $\cM$ induces, via the pushforward map, a linear transformation of $\cK^{p}(\cM)$ which defines a representation of $G$ \cite{MMS2004}. 
Moreover, the action $G \circlearrowright \cK^{2}(\cM)$ is not transitive. These observations are the crux of the invariant theory of Killing tensors, allowing us to formulate the canonical forms and
equivalence problems for the Killing tensors defined on spaces of constant curvature.

Of particular importance for applications are the elements of $\cK^{2}(\cM)$ enjoying additional geometric and algebraic properties. More specifically, let $\vec{K} \in \cK^{2}(\cM)$ be such that its eigenvalues are pointwise simple and real, and the eigenvector fields are  normal (orthogonally integrable). Such a Killing tensor $\vec{K}$ is called a {\em characteristic Killing tensor} (CKT) \cite{E34, Be97}. The remarkable property of CKTs is that their eigenvalues and eigenvectors generate
OSWs, which are $n$ foliations of the space $\cM$ that consist of $(n-1)$-dimensional hypersurfaces orthogonal to the eigenvectors of the CKT in question.

The action $G \circlearrowright \cK^{2}(\cM)$, thus defined, foliates the vector space $\cK^{2}(\cM)$ into the orbit space $\cK^{2}(\cM)/G$. These orbits represent  Killing tensors of valence two that share the geometric (and algebraic) properties that are equivalent modulo the group action $G \circlearrowright \cK^{2}(\cM)$. Of particular importance in applications are the orbits  that correspond to the CKTs belonging to the vector space $\cK^{2}(\cM)$. It must be noted at this point that, in general, the topology of the orbit space $\cK^{2}(\cM)/G$ (or $(\cK^{2}(\cM)\times \cM)/G$) is far from being trivial which makes the problem of {\em invariant classification} of the orbits fairly complicated. The problem consists of two ``subproblems," namely the {\em canonical forms problem} and the {\em equivalence problem} which can briefly be formulated in this setting as follows.
\begin{enumerate}
\item {\bf  Canonical forms problem:} Consider the action $G \circlearrowright \cK^{2}(\cM)$. The problem is to determine the number of  {\em inequivalent} orbits corresponding to the CKTs defined on $(\cM, \vec{g})$ as well as   the canonical forms representing each of them.

\item {\bf Equivalence problem:}  Consider again the action $G \circlearrowright \cK^{2}(\cM)$. Let $\vec{K} \in \cK^{2}(\cM)$. First, the problem is  to deterimine whether or not $\vec{K}$ is a CKT. If the answer is ``yes'', the main problem is to determine the corresponding   orbit in the quotient space $\cK^{2}(\cM)/G$ that the Killing two-tensor in question $\vec{K}$ belongs to.  Finally, we also want to determine the {\em moving frames map} \cite{FO98, FO99} that maps $\vec{K}$ to its respective canonical form.

\end{enumerate}
Recall that Eisenhart \cite{E34} outlined the solution to the canonical forms problem for CKTs formulated above and solved it for the case $\cM = \Eset^3$. In \cite{HMS09} (see also the relevant references therein) Horwood {\em et al} reformulated Eisenhart's approach in the language of the Cartan geometry (which Eisenhart employed implicitly). Note that the canonical forms problems for the cases $\cM = \Sset^3$ and $\cM = \Hset^3$ were solved by Olevsky in \cite{O50}, while the case $\cM = \Mset^3$ was treated by Horwood and McLenaghan in \cite{HM07}. Moreover, these solutions have been used to solve the corresponding equivalent problems. Thus, Horwood {\em et al} \cite{HMS05} and Horwood \cite{JTH07} employed two different methods to solve the equivalence problem for the CKTs defined on $\Eset^3$, while Horwood {\em at al} \cite{HMS09} solved the equivalence problem for the case $\cM = \Mset^3$.

In order to set the stage  for our theory and in what follows  solve the equivalence problem for $\cK^{2}(\cS)$, we now employ  the fundamental ideas from Cartan's approach to geometry. Observe first that
  $\cM \simeq G/H$, where $G$ is the isometry group of $\cM$ and $H$ is a closed subgroup of $G$. Thus, for example, $\Sset^3 \simeq
  SO(4)/SO(3)$. Next, we observe that in view of this identification the homogeneous space $G/H$ can be treated as the base manifold in the tautological principal bundle projection $\pi_{1}: G \rightarrow G/H \simeq \cM$ (i.e.\ $G$ is the principal $H$-bundle over $G/H$). Similarly,
  consider the vector bundle $\mathcal{K}^p(\cM)\times \cM$, then  $\pi_2 \circ \pi_{2} : \cK^{p}(\cM) \times \cM \rightarrow \cM \simeq G/H
$ is the vector bundle projection with the same base manifold. Upon noticing that  the transitive action of the isometry group  $I(M) = G$
in $\cM$ yields the non-transitive action of $G$ in  $\mathcal{K}^p(\cM)$, where $G$ acts as an automorphism, we consider next the orbit space
$(\mathcal{K}^p(\cM)\times \cM)/G$ leading to the third projection $ \pi_{3}: \cK^{p}(\cM) \times \cM \rightarrow (\cK^{p}(\cM) \times \cM) / G$,
having the structure of a principle $G$-bundle with $\mathcal{K}^p(\cM)\times \cM$ as the total  space. Following \cite{HMS09}, we introduce
the {\em lift} $f: (\mathcal{K}^2(\cM)\times \cM)/G \rightarrow G$, so that   the following diagram commutes:

\begin{equation}
  \xymatrix{
    G \ar[r]^-{\pi_{1}} & G/H \simeq \cM &  \\
    (\cK^{p}(\cM) \times \cM) / G \ar[u]^-{f}
      & \cK^{p}(\cM) \times \cM \ar[l]^-{\pi_{3}} \ar[u]_-{\pi_{2}}
       &
  } \label{D1CartanGeometry}
\end{equation}

Geometrically, the existence of $f$ is equivalent to the existence of a cross-section through the orbits  of the orbit-space $(\mathcal{K}^p(\cM)\times \cM)/G$. The intersections of such a cross-section with the orbits are the corresponding {\em canonical forms}, and their coordinates are {\em covariants (invariants)} which play an important role in the considerations that follow.

Consider now the case of $p=2$.

Let
$\vec{K} \in \cK^{2}(\cM)$ be a CKT at a non-singular point $\vec{x} \in \cM$ (i.e.\ the
eigenvalues of $\vec{K}$ are all real and distinct at $\vec{x}$). Indeed,
$\vec{K}$ gives rise to a quasi-orthonormal frame $E_{\vec{K}, \vec{x}}(\cM)$ of
eigenvectors $\{\vec{e}_{1}, \ldots, \vec{e}_{n}\}$ of $\vec{K}$ at $\vec{x}
\in \cM$, which is also a quasi-orthonormal basis for $T_{\vec{x}}(\cM)$.
Denoting $E(\cM)$ as the corresponding bundle of frames generated by CKTs in
$\cM$, it follows that $(E(\cM), \cM, \tilde{\pi}_{2})$ defines an (oriented)
quasi-orthonormal frame bundle, where $
  \tilde{\pi}_{2}: E(\cM) \rightarrow \cM.$ The fibres $\tilde{\pi}_{2}{}^{-1}(\vec{x})$ correspond to sets of all possible quasi-orthonormal frames at $\vec{x} \in \cM$ generated by the eigenvectors of CKTs. Finally, this arrangement leads to   the fibre bundle projection $
  \tilde{\pi}_{3} : \cK^{2}(\cM) \times \cM \rightarrow E(\cM).$  Accordingly,  our diagram (\ref{D1CartanGeometry}) now assumes the following form:

 \begin{equation}
  \xymatrix{
    G \ar[r]^-{\pi_{1}} & G/H \simeq \cM & \\
    (\cK^{2}(\cM) \times \cM) / G \ar[u]^-{f}
      & \cK^{2}(\cM) \times \cM \ar[l]^-{\pi_{3}} \ar[u]_-{\pi_{2}}
      \ar[r]_-{\tilde{\pi}_{3}} & E(\cM) \ar[ul]_-{\tilde{\pi}_{2}}
  } \label{D2CartanGeometry}
\end{equation}
The existence of such a lift $f$ is assured by the fact that $G$ acts
transitively on the bundle of frames $E(\cM)$ for a given CKT $\vec{K} \in
\cK^{2}(\cM)$ and the classical Cartan lemma \cite{PG74}:
\begin{lemma}[Cartan] \label{lemma:Cartan}
  Suppose that $\varphi$ is a $\mathfrak{g}$-valued one-form on a connected (or
  simply connected) manifold $\cM$. Then there exists a $C^{\infty}$ map
  $F : \cM \rightarrow G$ with $F^{\ast} \omega = \varphi$ iff
  $$
    \rmd \varphi = \varphi \wedge \varphi,
  $$
  where $\omega$ is the Maurer-Cartan form on $G$.
  Moreover, the resulting map is unique up to left translation.
\end{lemma}
    Furthermore, we  define the map $F : E(\cM) \rightarrow
G$ to be $F = f \circ \pi_{3} \circ \tilde{\pi}_{3}{}^{-1}$. Clearly, in this case $G$ may be identified
with the set of frames associated to the CKT $\vec{K}$ in question, or alternatively, cross-sections of the fibration
$G \rightarrow G/H$ over $\vec{K}$. In view of Lemma \ref{lemma:Cartan} the Maurer-Cartan form on $G$ can be restricted
to this choice of frames and thus produce a complete set of invariants (covariants) to solve the equivalence problem, which
one can now solve  for the orbit space $(\cK^{2}(\cM) \times \cM) / G$ (or $\cK^{2}(\cM) /
G$) for CKTs using the classical calculus of differential forms. More
specifically, the problem of invariant classification of the orbit(s)
generated by a CKT $\vec{K} \in \cK^{2}(\cM)$ reduces to fixing a
quasi-orthonormal frame of eigenvectors $\{\vec{e}_{1}, \ldots, \vec{e}_{n}\}$
and considering {\em in the frame\/} the corresponding Cartan structure
equations
\begin{align}
 & \rmd \vec{e}^{a} + \vec{\omega}^{a}{}_{b} \wedge \vec{e}^{b} = \vec{T}^{a},
     \label{eq:HJ:structeqn1} \\
 & \rmd \vec{\omega}^{a}{}_{b} + \vec{\omega}^{a}{}_{c} \wedge
   \vec{\omega}^{c}{}_{b} = \vec{\Theta}^{a}{}_{b}, \label{eq:HJ:structeqn2}
\end{align}
together with the Killing tensor equations for the components $K_{ab}$ of
$\vec{K}$ (\ref{CKTeqn})
and the integrability conditions
\begin{equation}\label{int_cond}
  \vec{e}^{a} \wedge \rmd \vec{e}^{a} = 0 \quad \text{(no sum)}.
\end{equation}
In these equations, $\vec{\omega}^{a}{}_{b} = \Gamma_{cb}{}^{a}\, \vec{e}^{c}$
are the connection one-forms, $\vec{T}^{a} = \frac{1}{2} T^{a}{}_{bc}\,
\vec{e}^{b} \wedge \vec{e}^{c}$ are the torsion two-forms,
$\vec{\Theta}^{a}{}_{b} = \frac{1}{2} R^{a}{}_{bcd}\, \vec{e}^{c} \wedge
\vec{e}^{d}$ are the curvature two-forms, $\{\vec{e}^{1}, \ldots,
\vec{e}^{n}\}$ is the dual basis of one-forms, the connection coefficients
$\Gamma_{cb}{}^{a}$ correspond to the Levi-Civita connection $\nabla$ (and
hence $\vec{T}^{a} = 0$ in (\ref{eq:HJ:structeqn1})) and $R^{a}{}_{bcd}$ are
the components of the curvature tensor. Note that (\ref{int_cond}) are the integrability
conditions for the normality of the eigenvectors $\vec{e}_a$.
We also note that, with respect to this
frame, the components of the metric $\vec{g}$ and CKT $\vec{K}$ are given by
\begin{equation}\label{gKframe}
  g_{ab} = \mathrm{diag}(\epsilon_{1}, \ldots, \epsilon_{n}), \quad
  K_{ab} = \mathrm{diag}(\epsilon_{1} \lambda_{1}, \ldots, \epsilon_{n}
    \lambda_{n}),
\end{equation}
respectively, where $\epsilon_{a} = \pm 1$, $a = 1, \ldots, n$, and
$\lambda_{a}$, $a = 1, \ldots, n$, are the eigenvalues of $\vec{K}$. The
{\em differential invariants\/} characterizing the orbits in question are
determined by the connection one-forms $\vec{\omega}^{a}{}_b$ that are
found from a fixed quasi-orthonormal frame $\{\vec{e}_{1}, \ldots,
\vec{e}_{n}\}$. More specifically, solving the Killing tensor equation (\ref{CKTeqn}) for $K_{ab}$ in this case modulo the integrability
conditions (\ref{int_cond}) will produce a set of canonical forms (thus solving the canonical forms problem) corresponding to the orbits, while finding the connection one forms $\vec{\omega}^{a}{}_b$ will provide a means to distinguish between the orbits, thus solving the equivalence problem (or, rather, most of it).  However, at this point we would like to point out that solving the equivalence problem
in this way (i.e., by thus finding a complete set of differential invariants) may be extremely challenging computationally. Instead, one can make use of the fact that the group $G$ acts transitively on the bundle of frames and  try to solve it ``in the group", employing an algebraic approach to Cartan's method of moving frames \cite{FO98,FO99,Olv99}. More specifically, via the tensor transformation laws, we can determine the action of $G$ on the parameters $a_1, \ldots, a_d$ (and the local coordinates of $\cM$ if necessary) that determine the vector space  $\cK^{2}(\cM)$ (or the product space $\cK^{2}(\cM)\times \cM$)  and then find a set of complete invariants (covariants) as algebraic functions of the parameters (parameters and local coordinates). This approach (analogous to the classical invariant theory of homogeneous polynomials) is also more preferable from the applications point of view, because the Killing two-tensors in Classical Mechanics (for example) normally appear in terms of local coordinates on $\cM$, rather than in the frame of its eigenvectors as in (\ref{gKframe}).

 The concept of a {\em web symmetry}, which signifies that the web is
invariant under at least a one-parameter group of isometries, is another
interesting geometric consequence stemming from the properties of covariants
(eigenvalues) of CKTs. Suppose, for example, $\vec{K}$ is a CKT in
$\cK^{2}(\cM)$ having $n-1$ functionally independent eigenvalues $\lambda_{i}$,
$i = 1, \ldots, n-1$. It follows that there exists a vector field $\vec{V} \in
\cK^{1}(\cM)$ ($\vec{V}$ is necessarily a Killing vector), such that
$\mathcal{L}_{\vec{V}}(\vec{K}) = 0$. Indeed, the integral curves of $\vec{V}$
are given by the common level sets
$$
  \bigcap_{i=1}^{n-1} \{\lambda_{i} = \text{const}\}.
$$

Now we can be more specific about solving the equivalence problem outlined above. Let  $\vec{K} \in \cK^{2}(\cM)$. We first verify whether or not the Killing two-tensor $\vec{K}$ in question is a CKT.  This requires verifying that the eigenvalues of $\vec{K}$ are real and distinct, and that the eigenvector fields of $\vec{K}$ are normal.
The latter step can be accomplished by verifying the vanishing of the {\em Haantjes tensor} ($H$ condition)
\cite{H51} (see \cite{H08} for more details):

\begin{equation}
\label{Haantjes1}
\vec{H}_{\hat{\vec{K}}} (X, Y):= \hat{\vec{K}}^2\vec{N}_{\hat{\vec{K}}}(X,Y) + \vec{N}_{\hat{\vec{K}}}(\hat{\vec{K}}X,\hat{\vec{K}}Y) - \hat{\vec{K}}(\vec{N}_{\hat{\vec{K}}}(X,\hat{\vec{K}}Y) +
\vec{N}_{\hat{\vec{K}}}(\hat{\vec{K}}X,Y))=0,
\end{equation}
where the $(1,1)$-tensor $\hat{\vec{K}} = \vec{K}\vec{g}^{-1}$, $\vec{N}_{\hat{\vec{K}}}$ is the Nijenhuis tensor \cite{N51} and $X, Y$ are arbitrary vector fields on $\cM$. Thus, the vanishing of the Haantjes tensor defined by (\ref{Haantjes1}) is equivalent to the Killing tensor $\vec{K}$ being a CKT. Alternatively, the formula (\ref{Haantjes1}) can be given in index form as follows:

\begin{equation}
\label{Haantjes2}
H^{i}{}_{jk}  = N^{i}{}_{\ell m}{K}^{\ell}{}_{j}{K}^{m}{}_{k} +2N^{\ell}{}_{m[j}{K}^{m}{}_{k]}{K}^{i}{}_{\ell} +  N^{\ell}{}_{jk}{K}^{m}{}_{\ell}{K}^{k}{}_{m}=0,
\end{equation}
where $K^{i}{}_{j}$ denotes the components of the $(1,1)$-tensor $\vec{\hat{K}}$ and $[\ldots]$ denotes skew symmetrization over the enclosed indices.

One can also employ the Tonolo-Schouten-Nijenhuis (TSN) conditions to verify the normality of the eigenvectors of a symmetric tensor field $\vec{K}$ of valence (0,2). Indeed, $\vec{K}$ with real distinct eigenvalues has normal eigenvectors iff the following conditions are satisfied:
\begin{align}\begin{split}
  & N^{\ell}{}_{[ij} g_{k]\ell} = 0, \\
  & N^{\ell}{}_{[ij} K_{k]\ell} = 0, \\
  & N^{\ell}{}_{[ij} K_{k]m} {K}^{m}{}_{\ell} = 0,
\end{split}\label{eq:ITKT:tsn}\end{align}
where $N^{i}{}_{jk}$ are the components of the Nijenhuis tensor \cite{N51} of
$K^{i}{}_{j}$ defined by
\begin{equation}
  N^{i}{}_{jk} = K^{i}{}_{\ell} K^{\ell}{}_{[j,k]} + K^{\ell}{}_{[j}
    K^{i}{}_{k],\ell} . \label{eq:ITKT:nijenhuis}
\end{equation}
In the index-free form, formula (\ref{eq:ITKT:nijenhuis}) can be obtained from  the formula (\ref{Haantjes1}) by replacing the operator $N_{\hat{\vec{K}}}(\cdot, \cdot)$ with the Lie bracket $[$ , $]$.
\begin{remark}
If the eigenvalues of $\vec{K}$ are real and distinct, the H condition and the TSN conditions are equivalent.  If the multiplicity of some the eigenvalues of $\vec{K}$ is greater than one,
the TSN conditions no longer apply.  However, in this case the following more general result \cite{H51} needed in Section 4 is true:  {\em the distributions defined by the eigenspaces of $\vec{K}$ in a Riemannian space are orthogonally integrable if and only if
the H condition is satisfied.} Recall that {\em orthogonal integrability} signifies that the distribution defined by the orthogonal complement of the
eigenspace is integrable.
\end{remark}

If $\vec{K} \in \cK^{2}(\cM)$ is a CKT, our next step is to determine which orbit in $\cK^{2}(\cM)/G$ it belongs to, noting that the corresponding canonical form for the orbit in question is known from solving the canonical forms problem. This problem can in principle be solved by employing a number of mathematical tricks and tools, although it is hard in general due to the fact that the group action
$G \circlearrowright \cK^{2}(\cM)$ is not regular and thus the orbits have different dimensions.

To solve the problem, one can employ {\em algebraic invariants} of the group action $G \circlearrowright \cK^{2}(\cM)$.  These are the invariants of  the group $G$ acting in the {\em parameter space} of the vector space $\cK^{2}(\cM)$, or more specifically, functions of the parameters that remain unchanged under the induced action of the group $G$ (see, for example, \cite{HMS05,  WF65} for more details). It must be noted however, that the problem can normally be solved in relatively simple cases using only algebraic invariants, that is when one has to deal with few types of orbits (e.g., $\dim\cM$ is small). Recall that algebraic invariants of the CKTs were first introduced by Winernitz and Fri\v{s} \cite{WF65} and later rediscovered in McLenaghan {\em et al} \cite{MST02} in a different context. More generally, one has to recover more information about the orbits and their degeneracies. One way to deal with the problem is to ``improve'' the group action by considering it in the extended space $\cK^{2}(\cM) \times \cM$. After appropriate computations, this yields {\em algebraic covariants} of Killing tensors, which are functions of both the parameters of the vector space $\cK^{2}(\cM)$ and coordinates of $\cM$. The covariants of Killing tensors were introduced by Smirnov and Yue in \cite{SY04} and have been successfully employed to solve a number of equivalence problems (see, for example, Horwood \cite{JTH07}). The degeneracies of the orbits in $\cK^{2}(\cM)/G$ (or $(\cK^{2}(\cM)\times \cM)/G)$ manifest themselves in the orbits having different dimensions - a fact which for our problem  has appropriate group theoretical, geometric and algebraic interpretations. Thus recall that in general  the dimension of an orbit is determined by the formula
$\mathrm{dim}\, G = \mathrm{dim}\, O_{\vec{x}} + \mathrm{dim}\, G_{\vec{x}}$, where $G$ is a group acting on $\cM$, $\vec{x} \in \cM$, $O_{\vec{x}}$ is the
orbit through $\vec{x}$, and $G_{\vec{x}}$ is the isotropy subgroup of $G$
through $\vec{x}$. The existence of isotropy subgroups indicates the degeneracies of the orbits (i.e., their dimensions drop), which in our case of $\cM = \cK^{2}(\cM)/G$ from the geometric perspective
is equivalent to the CKTs corresponding to the degenerate orbits admitting {\em web symmetries} (see below). Algebraically this means that the CKTs corresponding to the degenerate orbits have functionally {\em dependent} eigenvalues. Note that if all the eigenvalues of a CKT $\vec{K}$ are functionally {\em independent} they can be taken as the orthogonal coordinates (see \cite{HMS09} for more details). Degeneracies of the orbits can also be characterized by {\em singular points} of the CKTs in question, that is the points where the eigenvalues coincide.  Unfortunately, most of the orthogonal coordinate webs used in applications come from the CKTs with degeneracies.

Solving the equivalence and canonical forms problems described above forms the mathematical foundation for applications arising in Classical Mechanics, in particular in the Hamilton-Jacobi theory of orthogonal separation of variables for the natural Hamiltonians defined in spaces of constant curvature (see \cite{HMS09} for relevant references and details). Let a Hamiltonian system be defined by a natural Hamiltonian of the form:
\begin{equation}
\label{Ham}
H(\vec{q}, \vec{p}) = \frac{1}{2}g^{ij}(\vec{q})p_ip_j + V(\vec{q}).
\end{equation}
The following theorem (Benenti \cite{Be97}), which is a generalization of the corresponding result due to Eisenhart concerning geodesic Hamiltonians \cite{E34}, establishes an important link between orthogonal separation of variables in the associated Hamilton-Jacobi equation for the Hamiltonian system defined by (\ref{Ham}).
\begin{theorem}
\label{Be}
A Hamiltonian system defined by (\ref{Ham}) is orthogonal separable iff there exists a CKT $\vec{K}$ such that
\begin{equation}
\label{compatibility}
\mbox{d}(\hat{\vec{K}}\mbox{d}V) = 0,
\end{equation}
where the $(1,1)$-tensor $\hat{\vec{K}} = \vec{K}\vec{g}$
\end{theorem}
In light of Theorem \ref{Be} the canonical and equivalence problems formulated above can be reformulated in the language of orthogonal separation of variables for natural Hamiltonian systems defined by (\ref{Ham}) in spaces of constant cutvature as follows:
\begin{itemize}
  \item[(i)] How many ``inequivalent'' coordinate systems afford orthogonal
    separation of variables in the corresponding HJ equation?
  \item[(ii)] If the answer to (i) is non-zero, how can one characterize
    intrinsically the coordinate systems that afford separation of variables in
    the HJ equation?
  \item[(iii)] What are the canonical coordinate transformations
    $$
      (q^{1}, q^{2}, \ldots, q^{n}) \rightarrow  (u^{1}, u^{2}, \ldots, u^{n})
    $$
    from the given position coordinates of (\ref{Ham}) to the
    coordinate systems that afford orthogonal separation of variables of the
    HJ equation?
\end{itemize}

In what follows we employ some of the tools and techniques described  above to  solve the equivalence problem for the CKTs defined on  $\Sset^3$, employing the results obtained by Olevsky \cite{O50} and Kalnins {\em et al} \cite{KMW76} who solved the corresponding canonical forms problem.


\section{Killing two-tensors on $\Sset^n \subset \Eset^{n+1}$}

To fully explore the algebraic and geometric properties of Killing two-tensors defined in a sphere, we employ the (Cartesian) coordinates of the corresponding ambient Euclidean space. In this view the vector space of Killing two-tensors ${\cal K}^2 (\Sset^n)$ is viewed as a subspace of
${\cal K}^2(\Eset^{n+1})$ whose elements are defined in terms of the Cartesian coordinates of $\Eset^{n+1}$.

Indeed, since every Killing tensor defined in a space of constant curvature  is expressible as a sum of symmetrized products of Killing vectors, let us begin by defining basic Killing vectors of a Euclidean space $\Eset^m$ in terms of the corresponding Cartesian coordinates.  For our purposes, we  also need to define the dilatational vector field $\vec{D}$ (Euler vector field).  Thus, the translational $\vec{X}_i$ and rotational $\vec{R}_{ij}$ Killing vectors and the vector $\vec{D}$ in terms of a system of Cartesian coordinates $x^i$ can be defined as follows:
\begin{equation}
  \vec{X}_{i} = \frac{\partial}{\partial x^{i}}, \quad
  \vec{R}_{ij} = 2\delta^{k\ell}_{ij}g_{\ell m}x^{m} \vec{X}_{k}, \quad
  \vec{D} = x^{i}\vec{X}_{i},
    \label{KVbasis}\end{equation}
where $g_{ij}=\mbox{diag}(1,\ldots,1)$, denotes the Euclidean metric of the Euclidean space in question and $\delta^{k\ell}_{ij}=\delta^{k}_{[i}\delta^{\ell}_{j]}$, the generalized Kronecker delta. Note that the translational and rotational Killing vectors
defined above form a basis for the Lie algebra ${\cal K}^1(\Eset^n)$ of the group of rigid motions of $\Eset^m$.  Next, the commutation relations among these vectors are given by
\begin{align}\begin{split}
   [\vec{X}_i,\vec{X}_j] = 0, \quad
   [\vec{X}_i,\vec{R}_{jk}] =  2\delta^{\ell m}_{jk}g_{mi}\vec{X}_{\ell}, \quad
   [\vec{X}_i,\vec{D}] = \vec{X}_i, \\
   [\vec{R}_{ij},\vec{R}_{k\ell}] = 4\delta^{mn}_{ij}\delta^{pr}_{k\ell}g_{mp}\vec{R}_{nr}, \quad
   [\vec{D},\vec{R}_{ij}] = 0.
\end{split}\label{comrels}\end{align}
The Killing vectors also satisfy the following algebraic identities ({\em syzgies}):
\begin{equation}
   \vec{X}_{[i}\odot\vec{R}_{jk]} = 0, \quad
   \vec{R}_{i[j}\odot\vec{R}_{kl]} = 0,
     \label{ident}
\end{equation}
where $\odot$ denotes the symmetric tensor product.
The general Killing vector thus has the form
\begin{equation}
\vec{K}=A^{i}\vec{X}_{i} + B^{ij}\vec{R}_{ij},
\label{KV}\end{equation}
where the constants $A^{i}$, and $B^{ij}$, called the {\em Killing vector parameters}, satisfy the symmetry relation
\begin{equation}
B^{(ij)}=0.
\label{symrels1}\end{equation}
From (\ref{KVbasis}) and (\ref{KV}) we obtain the covariant components of $\vec{K}$ with respect to the natural basis $dx^i$:
\begin{equation}
K_i=A_i+2B_{ij}x^j.
\label{KVcomp}\end{equation}

Similarly, the general valence two Killing tensor thus has the form
\begin{equation}
   \vec{K} = A^{ij}\vec{X}_i\odot\vec{X}_j + B^{ijk}\vec{X}_i\odot\vec{R}_{jk} + C^{ijk \ell}\vec{R}_{ij}\odot\vec{R}_{k\ell},
\label{v2KT}
\end{equation}
where the constants $A^{ij}$, $B^{ijk}$, and $C^{ijk\ell}$, called the {\em Killing tensor parameters} satisfy the following symmetry relations
\begin{align}\begin{split}
   A^{[ij]} = 0, \\
   B^{i(jk)} = 0, \quad
   B^{[ijk]} = 0, \\
   C^{(ij)k\ell} = 0, \quad
   C^{ijk\ell} = C^{k\ell ij}, \quad
   C^{i[jk\ell]} = 0.
\end{split}\label{symrels}\end{align}
A key  observation is  that the parameter set  $C^{ijk\ell}$ has the same symmetries as the Riemann curvature tensor.  For this reason it is sometimes called an {\em algebraic curvature tensor}.  
From (\ref{KVbasis}) and (\ref{v2KT}) we find that the covariant components of $\vec{K}$ with respect to the natural basis are given by
\begin{equation}
   K_{ij} = A_{ij} + 2B_{(ij)k}x^k + 4C_{ikj\ell}x^{k}x^{\ell}.
\label{KTcomps}\end{equation}
This result is consistent with that obtained in \cite{MMS2004} obtained by the use
of representation theory.

Now we  may determine the  form of the general Killing tensor on $\Sset^n$ by considering it as an imbedded hypersurface in $\Eset^{m}$ with $m=n+1$ defined implicitly by the equation
\begin{equation}
\vec{x}\cdot\vec{x} = 1,
\label{nsph}\end{equation}
where $\vec{x}=(x^1,\ldots,x^{n+1})$, and $\cdot$ denotes the Euclidean inner product.
Let
\begin{equation}
\vec{x} = \vec{f}(u^1,\ldots,u^n),
\label{parnsph}
\end{equation}
be a local parametrization of $\Sset^n$.  Then from (\ref{nsph}) and (\ref{parnsph}) we obtain by differentiation
\begin{equation}
\vec{x}\cdot\vec{x}_{\alpha} = 0, \quad
\vec{x}_{\alpha\beta}\cdot\vec{x} = - {g}_{\alpha\beta},
\label{normal}\end{equation}
where $\vec{x}_{\alpha}= \frac{\partial\vec{x}}{\partial u^{\alpha}}$, ($x^{i}_{\alpha}=\frac{\partial x^i}{\partial u^{\alpha}}$), and
\begin{equation}
{g}_{\alpha\beta} = \vec{x}_{\alpha}\cdot\vec{x}_{\beta},
\label{1stfundform}\end{equation}
denotes the pullback of the Euclidean metric to $\Sset^n$. 
The characterization is given by the following proposition:
\begin{proposition}\label{delongprop}
The general Killing vector and general valence two Killing tensor fields on $\Sset^n$ are given by
\begin{equation}
K_{\alpha}=x^i_{\alpha}K_i, \quad K_{\alpha\beta} = x^{i}_{\alpha}x^{j}_{\beta}K_{ij},
\label{nspKT1}\end{equation}
where
\begin{equation}
K_i=2B_{ij}x^j, \quad K_{ij}=4C_{ikj\ell}x^{k}x^{\ell}.
\label{nspKT2}\end{equation}
\end{proposition}
\begin{proof}
The pullback to $\Sset^n$ of the covariant derivative of any smooth tensor field $K_{i_1,\ldots, i_p}$ defined on $\Eset^{n+1}$ is given by
\begin{equation}
x^{i_1}_{\alpha_1},\ldots,x^{i_p}_{\alpha_p}x^{i_{p+1}}_{\alpha_{p+1}}K_{i_1,\ldots,i_p,i_{p+1}}=K_{\alpha_1,\ldots,\alpha_p;\alpha_{p+1}}
+px^iK_{i(\alpha_2,\ldots,\alpha_{p-1}}g_{\alpha_p)\alpha_{p+1}}.
\label{pb1}\end{equation}
If $K_{i_1,\ldots,i_p}$ is any solution of the KT equation (\ref{CKTeqn}), it follows that
\begin{equation}
K_{(\alpha_1,\dots,\alpha_p;\alpha_{p+1})}+px^iK_{i(\alpha_1,\ldots,\alpha_{p-1}}g_{\alpha_p\alpha_{p+1})}=0.
\label{pb2}\end{equation}
We conclude that $\vec{K}$ will define a Killing tensor field on $\Sset^n$ if and only if
\begin{equation}
x^iK_{i(\alpha_1,\ldots,\alpha_{p-1}}g_{\alpha_p\alpha_{p+1})}=0.
\label{nasc1}\end{equation}
This condition is clearly satisfied if
\begin{equation}
x^jK_{ji_{1},\ldots,i_{p-1}}= 0,
\label{sc1}\end{equation}
on $\Eset^{n+1}$. For the cases $p=1,2$ this condition takes the form
\begin{equation}
x^iK_i=0, \quad x^jK_{ij}=0.
\label{sc2}\end{equation}
It follows from (\ref{symrels1}), (\ref{KVcomp}), (\ref{symrels}), and (\ref{KTcomps}) that
\begin{equation}
A_i=0, \quad A_{ij}=0, \quad B_{(ij)k=0}.
\label{sc3}\end{equation}
These conditions imply that $K_i$ and $K_{ij}$ have form given by (\ref{nspKT2}).  We note that, with the use of (\ref{comrels}), the conditions
(\ref{sc3}) may be written invariantly as
\begin{equation}
[\vec{D},\vec{K}]=0.
\label{sc4}\end{equation}

It remains to demonstrate that (\ref{nspKT1}) and (\ref{nspKT2}) indeed define the general Killing vector and general valence two Killing tensor on $\Sset^n$.
This may be achieved by a dimensional argument.  From (\ref{dimKT}) it follows respectively for $p=1,2$ that
\begin{equation}
d=\frac{1}{2}n(n+1), \quad \frac{1}{12}n(n+1)^2(n+2).
\label{KTdimnsp}\end{equation}
For the case $p=1$, the number $d$ is identical to the number of independent components of the Killing vector parameter $B_{ij}$,
while for the case $p=2$, $d$ is identical to the number of independent components of the algebraic curvature tensor $C_{ijk\ell}$ in a $(n+1)$-dimensional space.
\end{proof}
\begin{remark}\label{delong}
Proposition \ref{delongprop} is closely related to a result proved by Delong \cite{delong} which states that the valence $p$ Killing tensors (considered as functions on the
cotangent bundle of $\Eset^{n+1}$) that are in involution with the Euclidean distance function and the function $x^ip_i$, where $p_i$ denotes
the canonical momenta, are Killing tensors on $\Sset^n$.
Delong's conditions are equivalent to (\ref{sc2}) and (\ref{sc4}) respectively, only one of which seems to required to obtain the result.  Furthermore, our proposition goes
beyond his by giving the
explicit form of the general valence one and two KTs on $\Sset^n$ in terms of the appropriate KTs in the
ambient space $\Eset^{n+1}$.  Our proposition may be easily generalized to valence $p$ Killing tensors by the use of the results of Horwood \cite{JTH07}.
\end{remark}
\begin{remark}
Note that the image of the contravariant metric metric $g^{\alpha\beta}$ on $\Sset^{n}$ with respect to spherical coordinates under the pushforward induced by the parametrization 
gives rise to a Killing tensor $\vec{\mathcal{C}}$  called the {\em Casimir tensor}  with the property that it commutes with every element of the vector subspace of ${\cal K}^2(\Eset^{n+1})$ given by (\ref{nspKT2}). In view of the above  $\vec{\mathcal{C}} \in {\cal K}^2(\Eset^{n+1})$ is degenerate as an element of the vector subspace determined by $4C_{ikj\ell}x^{k}x^{\ell}$.  We also note that $\vec{\mathcal{C}}$ is a CKT for the spherical web on $\Sset^{n}$.  The explicit form of the Casimir tensor is given by
\begin{equation}\label{casimir1}
\vec{\mathcal{C}}=kg^{ik}g^{\ell j}\vec{R}_{ij}\odot\vec{R}_{k\ell},
\end{equation}
or with respect to the natural basis by
\begin{equation}\label{casimir2}
\mathcal{C}_{ij}=kg_{i[j}g_{\ell]k}x^kx^{\ell}
\end{equation}
where $k$ is some constant.  It may be shown by direct calculation that $[\vec{\mathcal{C}},\vec{K}]=0$, which verifies the above mentioned
commutation property.

\end{remark}

The next step is to determine the KTs among those given by (\ref{nspKT1}) and (\ref{nspKT2}) that are CKTs, that is which have point-wise distinct
eigenvalues and normal eigendirections.  To effect this determination we use the fact the eigenvalues of $K_{\alpha\beta}$ are also eigenvalues of
$K_{ij}$ and that the pushforward of the eigenvectors of  $K_{\alpha\beta}$ are eigenvectors of $K_{ij}$.  These results follow by contracting both
sides of the eigenvalue equation
\begin{equation}
K_{ij}X^{j}=\lambda g_{ij}X^{j},
\label{eveq1}\end{equation}
with $x^{i}_{\alpha}$, which yields
\begin{equation}
K_{\alpha\beta}X^{\beta}=\lambda g_{\alpha\beta}X^{\beta},
\label{eveq2}\end{equation}
where $X^{j}$ in (\ref{eveq1}) is the pushforward of $X^{\beta}$
\begin{equation}
X^{j}=x^{j}_{\beta}X^{\beta}.
\label{eveq3}\end{equation}
We conclude that if $K_{ij}$ has distinct eigenvalues then so does $K_{\alpha\beta}$.  The converse is not true since if $K_{\alpha\beta}$ has a zero eigenvector then by (\ref{sc2}) $K_{ij}$ will have zero as a repeated eigenvalue.  However, the remaining eigenvalues will be non-zero and distinct.  We also observe by (\ref{normal}) that the eigenvector $x^{i}$ corresponding the the zero eigenvalue of $K_{ij}$ pulls back to the zero vector on $\Sset^n$.

We next show that the pullback of a normal co-vector field on $\Eset^{n+1}$ is a normal vector field on $\Sset^n$.  To see this consider any
co-vector field $E_{i}$ on $\Eset^{n+1}$ which satisfies the normality condition
\begin{equation}
E_{[i,j}E_{k]}=0.
\label{normcond}\end{equation}
We note that (\ref{normcond}) is the component form of the integrability condition (\ref{int_cond}).
An easy computation shows that
\begin{equation}
E_{i,j}E_{k}x^{i}_{\alpha}x^{j}_{\beta}x^{k}_{\gamma}=E_{\alpha,\beta}E_{\gamma}-x^{i}_{\alpha\beta}E_{i}E_{\gamma}.
\label{pullb}\end{equation}
It follows from (\ref{normcond}) and (\ref{pullb}) that
\begin{equation}
E_{[\alpha,\beta}E_{\gamma]}=0,
\label{normcond2}\end{equation}
which implies that the pullback $E_{\alpha}$ is a normal vector field.
The above results imply that the pullback of a characteristic KT on $\Eset^{n+1}$ is a characteristic KT on $\Sset^{n}$.  In view of the above results
we are able to study CKTs on $\Sset^{n}$ as pullbacks of CKTs on $\Eset^{n+1}$.

We observe that the use of (\ref{int_cond})  to determine whether the eigenvector fields of $\vec{K}$ are normal, is
impractical since it requires the explicit determination of the eigenvector fields,  an intractable problem.  However, alternatively we may
utilize either the TSN conditions (\ref{eq:ITKT:tsn}) or the H condition defined by (\ref{Haantjes2}).  We first consider
the TSN conditions for a KT of the form (\ref{nspKT2}).
These conditions impose the corresponding  algebraic conditions on the algebraic curvature tensor $C$:
\begin{equation}
C^{\ell}{}_{(pq[i}C_{jk]r)\ell}=0,
\label{tsn1}\end{equation}
\begin{equation}
C_{\ell(pq}{}^{m}C^{\ell}{}_{r[ij}C_{k]st)m}-2C_{\ell(pq[i}C_{j}{}^{\ell}{}_{|r|}{}^{m}C_{k]st)m}=0,
\label{tsn2}\end{equation}
\begin{eqnarray}
3C_{\ell(pq}{}^{m}C^{\ell}{}_{r|n|s}C^{n}{}_{t[ij}C_{k]uv)m}+2C_{\ell(pq}{}^{m}C_{|n|rs[i}C_{j|t|}{}^{n\ell}C_{k]uv)m}+ \nonumber \\[0.3cm]
2C_{\ell(pq}{}^{m}C_{|n|rs[i}C_{j}{}^{n}{}_{|t|}{}^{\ell}C_{k]uv)m}=0,
\label{tsn3}\end{eqnarray}
where $|\ldots|$ denotes exclusion of the enclosed indices from the symmetrization process.
The first TSN condition (\ref{tsn1}) implies
\begin{equation}
C_{\ell(pq}{}^{m}C^{\ell}{}_{r[ij}C_{k]st)m}+2C_{\ell(pq[i}C_{j}{}^{\ell}{}_{|r|}{}^{m}C_{k]st)m}=0.
\label{tsn4}\end{equation}
This equation together with the second TSN condition (\ref{tsn2}) yields the two equations
\begin{equation}
C_{\ell(pq}{}^{m}C^{\ell}{}_{r[ij}C_{k]st)m}=0,
\label{tsn5}\end{equation}
\begin{equation}
C_{\ell(pq[i}C_{j}{}^{\ell}{}_{|r|}{}^{m}C_{k]st)m}=0.
\label{tsn6}\end{equation}
These equations replace the second TSN condition (\ref{tsn2}).
We now show that the third TSN condition is a consequence of the first and second conditions.  The cyclical identity implies
\begin{equation}
C_{\ell(pq}{}^{m}C_{|n|rs[i}C_{j}{}^{\ell}{}_{|t|}{}^{n}C_{k]uv)m}=-C_{\ell(pq}{}^{m}C_{|n|rs[i}C_{j|t|}{}^{n\ell}C_{k]uv)m}+C_{\ell(pq}{}^{m}C_{|n|rs[i}C_{j}{}^{n}{}_{|t|}{}^{\ell}C_{k]uv)m}.
\label{cyclical}\end{equation}
From the above equation and (\ref{tsn6}) we obtain
\begin{equation}
C_{\ell(pq}{}^{m}C_{|n|rs[i}C_{j|t|}{}^{n\ell}C_{k]uv)m}=0,
\label{tsn7}\end{equation}
and
\begin{equation}
C_{\ell(pq}{}^{m}C_{|n|rs[i}C_{j}{}^{\ell}{}_{|t|}{}^{n}C_{k]uv)m}=C_{\ell(pq}{}^{m}C_{|n|rs[i}C_{j}{}^{n}{}_{|t|}{}^{\ell}C_{k]uv)m}.
\label{tsn8}\end{equation}
Now from the first TSN condition and (\ref{tsn8}) we find
\begin{equation}
C_{\ell(pq}{}^{m}C^{\ell}{}_{r|n|s}C^{n}{}_{t[ij}C_{k]uv)m}-2C_{\ell(pq}{}^{m}C_{|n|rs[i}C_{j}{}^{n}{}_{|t|}{}^{\ell}C_{k]uv)m}=0.
\label{tsn9}\end{equation}
On the other hand (\ref{tsn5}) implies
\begin{equation}
C_{\ell(pq}{}^{m}C^{\ell}{}_{r|n|s}C^{n}{}_{t[ij}C_{k]uv)m}+2C_{\ell(pq}{}^{m}C_{|n|rs[i}C_{j}{}^{n}{}_{|t|}{}^{\ell}C_{k]uv)m}=0.
\label{tsn10}\end{equation}
These equations together imply
\begin{equation}
C_{\ell(pq}{}^{m}C^{\ell}{}_{r|n|s}C^{n}{}_{t[ij}C_{k]uv)m}=0,
\label{tsn11}\end{equation}
\begin{equation}
C_{\ell(pq}{}^{m}C_{|n|rs[i}C_{j}{}^{n}{}_{|t|}{}^{\ell}C_{k]uv)m}=0.
\label{tsn12}\end{equation}
Finally (\ref{tsn7}), (\ref{tsn11}), and (\ref{tsn12}) imply that the third TNS condition (\ref{tsn3}) is identically satisfied.
\begin{remark}
This result was first proven in $\Eset^3$ by Czapor \cite{HMS05} using Gr\"{o}bner basis theory and computer algebra.
Sch\"{o}bel \cite{S10} extended the result
to $n$-dimensional spaces of non-zero constant curvature using representation theory of the symmetric group.  However, it seems that our proof,
based on standard indicial tensor algebra, is simpler and more concise.
\end{remark}
Substituting the Killing tensor (\ref{nspKT2}) into the $H$ condition (\ref{Haantjes2}) we obtain the following condition on the coefficients:
\begin{eqnarray}
4C_{\ell(pq}{}^{k}C^{m}{}_{rs|i}C_{j|tu}{}^{n}C^{\ell}{}_{v)mn}+2C_{\ell(p|m|}{}^{k}C^{n}{}_{qr[i}C_{j]st}{}^{m}C^{\ell}{}_{uv)n} \nonumber \\
-5C_{\ell(pq}{}^{k}C^{m}{}_{rs[i}C_{j]|m|t}{}^{n}C^{\ell}{}_{uv)n}+C_{\ell(pq}{}^{k}C^{m}{}_{rs[i}C_{j]}{}^{\ell}{}_{t}{}^{n}C_{|n|uv)m} \nonumber \\
+C_{\ell(pq}{}^{k}C^{m}{}_{rs[i}C_{j]tu}{}^{n}C_{|n|v)m}{}^{\ell}-3C_{\ell(pq}{}^{k}C^{m}{}_{r|ij|}C^{n}{}_{st|m|}C^{\ell}{}_{uv)n}{} \nonumber \\
-2C_{\ell(pq}{}^{k}C^{m}{}_{rs[i}C_{j]t|m|}{}^{n}C^{\ell}{}_{uv)n}=0.
\label{h}\end{eqnarray}
Using this condition we can verify whether or not a given Killing tensor defined on $\Sset^n$ is characteristic. 
It may be shown that that the first and second TSN conditions (\ref{tsn1}) and (\ref{tsn2}) imply that the H condition (\ref{h}) is satisfied.

We next examine the pullback of the quadratic first integral
\begin{equation}
K=K^{ij}p_{i}p_{j}+U.
\label{quadint}\end{equation}
The vanishing of the Poisson bracket implies that $K_{ij}$ satisfies (\ref{CKTeqn}) for $p=2$, and that
\begin{equation}
\frac{\partial U}{\partial x^i}= g_{ij}K^{j\ell}\frac{\partial V}{\partial x^{\ell}},
\label{int2}\end{equation}
where $V$ is the potential function.
The pullback to $\Sset^{n}$ is given by
\begin{equation}
\frac{\partial U}{\partial x^i}x^{i}_{\alpha}=g_{ij}K^{\beta\gamma}x^{j}_{\beta}x^{\ell}_{\gamma}\frac{\partial V}{\partial x^{\ell}}x^{i}_{\alpha},
\end{equation}
which may be written as
\begin{equation}
\frac{\partial U}{\partial u^{\alpha}}=g_{\alpha\beta}K^{\beta\gamma}\frac{\partial V}{\partial u^{\gamma}}.
\end{equation}

We conclude this section by studying the action of $\SE(m)$ on the Killing vector and Killing tensor parameters \cite{JTH07}.  The action  $\SE(m) \circlearrowright \Eset^{m}$ is given by
\begin{equation}
x^i=\Lambda^{i}{}_{j}\tilde{x}^j+\delta^i,
\label{act1}\end{equation}
where $\Lambda^{i}{}_{j} \in \SO(m)$, $\delta^i \in \Eset^{m}$, and $\tilde{x}^i$ denote the transformed Cartesian coordinates.  This tranformation induces by (\ref{KVbasis}) the following transformation of the Killing vectors:
\begin{equation}
\vec{X}_i=\Lambda^{j}{}_{i}\vec{\tilde{X}}_j, \quad \vec{R}_{ij}=\Lambda^{k}{}_{i}\Lambda^{\ell}{}_{j}\vec{\tilde{R}}_{k\ell}+
\mu_{ij}{}^k\vec{\tilde{X}}_k,
\label{act2}\end{equation}
where
\begin{equation}
\mu_{ij}{}^{k}=2\delta^{\ell m}_{ij}g_{mn}\Lambda^{k}{}_{\ell}\delta^m.
\label{act3}\end{equation}
It follows that the Killing vector and parameters transform as
\begin{equation}
\tilde{A}^i=\Lambda^{i}{}_{j}A^j+\mu_{jk}{}^{i}B^{jk}, \quad \tilde{B}^{ij}=\Lambda^{i}{}_{k}\Lambda^{j}{}_{\ell}B^{k\ell},
\label{act4}\end{equation}
while the valence two Killing tensor parameters transform as
\begin{eqnarray}
\tilde{A}^{ij}=\Lambda^{i}{}_{k}\Lambda^{j}{}_{\ell}A^{k\ell}+2\Lambda^{(i}{}_{k}\mu_{\ell m}{}^{j)}B^{klm}+
\mu_{k\ell}{}^{i}\mu_{mn}{}^{j}C^{klmn}, \nonumber \\
\tilde{B}^{ijk}=\Lambda^{i}{}_{\ell}\Lambda^{j}{}_{m}\Lambda^{k}{}_{n}B^{\ell mn}+\mu_{mn}{}^{i}\Lambda^{j}{}_{p}\Lambda^{k}{}_{q}C^{mnpq},
\nonumber \\
\tilde{C}^{ijk\ell}=\Lambda^{i}{}_{m}\Lambda^{j}{}_{n}\Lambda^{k}{}_{p}\Lambda^{\ell}{}_{q}C^{mnpq}.
\label{act5}\end{eqnarray}
In order to make our formulas more compact we introduce a multi-index notation where any upper-case index represents a pair of skew-symmetric lower-case indices.
Thus $\vec{R}_I$ will represent $\vec{R}_{ij}$.  Using this notation we may rewrite (\ref{act2}) as
\begin{equation}
\vec{R}_I=\Lambda^{J}{}_{I}\vec{\tilde{R}}_J+\mu_{I}{}^{k}\tilde{X}_{k},
\label{act6}\end{equation}
where $\Lambda^{I}{}_{K}=\Lambda^{i}{}_{[k}\Lambda^{j}{}_{\ell]}$, represents the second compound of $\Lambda^{i}{}_{j}$.  With the same notation
(\ref{act4}) reads
\begin{equation}
\tilde{A}^i=\Lambda^{i}{}_{j}A^j+\mu_{J}{}^{i}B^J, \quad
\tilde{B}^I=\Lambda^{I}{}_{J}B^J,
\label{act7}\end{equation}
whereas (\ref{act5}) reads
\begin{eqnarray}
\tilde{A}^{ij}=\Lambda^{i}{}_{k}\Lambda^{j}{}_{\ell}A^{k\ell}+2\Lambda^{(i}{}_{k}\mu_{L}{}^{j)}B^{kL}+
\mu_{K}{}^{i}\mu_{L}{}^{j}C^{KL}, \nonumber \\
\tilde{B}^{iJ}=\Lambda^{i}{}_{\ell}\Lambda^{J}{}_{M}B^{\ell M}+\mu_{L}{}^{i}\Lambda^{J}{}_{M}C^{LM},
\nonumber \\
\tilde{C}^{IJ}=\Lambda^{I}{}_{L}\Lambda^{J}{}_{M}C^{LM}.
\label{act8}\end{eqnarray}

We are now in a position to give an invariant classification of the translational Killing vectors on $\Eset^m$ under the action of $\SE(m)$.  The
invariant needed to effect this classification is $\mathcal{I}_1=B_{ij}B^{ij}$.

A {\em translational Killing vector} in $\Eset^m$ is invariantly defined by the condition $\mathcal{I}_1=0$, which implies $B^{ij}=0$.  Thus by (\ref{KV}) a translational KV has the form
\begin{equation}
\vec{K}=A^i\vec{X}_i,
\label{tkv1}\end{equation}
where it is assumed that not all the $A^i$ are zero so that $\vec{K}$ is non-trivial.
It follows from the transformation formula (\ref{act4}) that, by an appropriate choice of $\Lambda^{i}{}_{j}\in \SO(m)$ \cite{HMS05}, we can transform
(\ref{tkv1}) to the form
\begin{equation}
\vec{K}=\tilde{A}^{1}\tilde{X}_{1},
\label{tkv2}\end{equation}
for some $\tilde{A}^1 \neq 0$.

The case when $\mathcal{I}_1 \neq 0$, characterizes the {\em non-translational Killing vectors} the analysis of which will be left to a future paper.

However, we are in a position to classify the rotational Killing vectors in $\Eset^4$ given by (\ref{nspKT2}) that define KVs on $\Sset^3$.  For this
we need the additional invariant $\mathcal{I}_2=\epsilon_{ijkl}B^{ij}B^{kl}$.
We utilize the property \cite{G67} that the skew-symmetric matrix $B^{ij}$, assumed  to be not the $0$ matrix, may be transformed to the following canonical form by an element of $\mathrm{O}(4)$:

\begin{equation}
  B^{ij} = \begin{pmatrix} 0 & \kappa_{1} & 0 & 0 \\ -\kappa_{1} & 0 & 0 & 0 \\
    0 & 0 & 0 & \kappa_{2} \\ 0 & 0 & -\kappa_{2} & 0 \end{pmatrix},
  \label{bij}
\end{equation}
where $\kappa_{1}, \kappa_{2} \in \Rset$.  If $\mathcal{I}_{1} \neq 0$ and $\mathcal{I}_{2}=0$, then $\kappa_{1}^2+\kappa_{2}^2 \neq 0$ and $\kappa_{1}\kappa_{2}=0$.  Without loss of generality we may set $\kappa_2=0$, which implies that the KV has the form
\begin{equation}
K=b_{12}R_{12},
\label{rot1}\end{equation}
where $b_{12}=2\kappa_{1}$ is some constant.  We now assume $\mathcal{I}_2 \neq 0$, which implies that $\kappa_1\kappa_2 \neq 0$.  In this case
the matrix $B^{ij}$ is rank four from which it follows that the KV has the form
\begin{equation}
K=b_{12}R_{12}+b_{34}R_{34},
\label{rot2}\end{equation}
where $b_{12}=2\kappa_1$ and $b_12=2\kappa_2$ are some constants.  The equations (\ref{rot1}) and (\ref{rot2}) via Proposition \ref{delongprop} give the two possible
canonical forms for the rotations on $\Sset^3$.


\section{The equivalence problem for Killing tensors on $\mathbb{S}^3$}
\label{section4}
To solve the equivalence problem outlined in the preceding section for the CKTs on $\Sset^3$ we employ as in the Cartesian coordinates of the ambient Euclidean space $\Eset^4 \supset \Sset^3$. In terms of these coordinates, the general Killing tensor of $\mathbb{S}^3$ is given by
\begin{equation}
\label{gK}
\vec{K} = 4C^{ijk\ell}\vec{R}_{ij}\odot \vec{R}_{k\ell},
\end{equation}
where
\begin{equation}
\label{genR}
\vec{R}_{ij} = \delta^{m}{}_{ijk}x^k\vec{X}_m, \quad \delta^{m}{}_{ijk} = \delta^m_ig_{jk} - \delta^m_jg_{ik},
\end{equation}
$$i,j,k,\ell, m = 1,\ldots, 4,$$
$x^1,\ldots, x^4$ are  Cartesian coordinates of $\Eset^4$, $\vec{X}_i = \partial_i$,  $g_{ij}$ are the components of the metric and $\odot$ is the symmetric tensor product. Note that (\ref{genR}) defines the six generators of $so(4)$, the Lie algebra of the  isometry group $SO(4)$ of $\Sset^3$ (no reflections). Therefore the RHS of (\ref{gK}) represents a vector space of Killing tensors on $\Sset^3$. Furthermore, the dimension $d$ of this space is determined by the tensor $C^{ijk\ell}$ which has the same symmetries as that of the Riemann curvature tensor, hence $d =20$.

The special orthogonal group $\SO(4)$ is a Lie subgroup of the orthogonal group $\mathrm{O}(4)$, consisting of all orthogonal matrices $\Lambda$ with positive unit determinant.  The transitive action of $\SO(4)$ on $\mathbb{E}^4$ can be specified by
\[x^i = \Lambda^i_{\ j}\tilde{x}^j, \]
where $\Lambda^i_{\ j} \in \SO(4)$ and $x^i$ denote Cartesian coordinates. This, in turn, induces the following transformation
\begin{equation} \label{act9} \vec{R}_{ij} = \Lambda^k_{\ i }\Lambda^{\ell}_{\ j}\vec{\tilde{R}}_{k\ell}, \ \ \ K^{ij} = \Lambda^i_{\ k}\Lambda^j_{\ \ell}\tilde{K}^{k \ell}\end{equation}
on the Killing vectors (\ref{genR}) and Killing tensors (\ref{gK}) of $\mathbb{E}^4$.  At the same time, this action induces the following transformations
\begin{equation} \label{act10} \tilde{B}^{ij} = \Lambda^i_{\ k}\Lambda^j_{\ \ell} B^{k \ell}, \ \ \ \ \tilde{C}^{ijk\ell} = \Lambda^i_{\ p} \Lambda^j_{\ q} \Lambda^k_{\ r} \Lambda^{\ell}_{\ s} C^{pqrs} \end{equation}
on the Killing vector and Killing tensor parameters.

To make our formulas more compact, we will once again adopt a multi-index notation (i.e., $\vec{R}_I$ will represent $\vec{R}_{ij}$) .  Using this notation we may rewrite the first equation of (\ref{act9}) as
\begin{equation}
\vec{R}_I=\lambda^{J}{}_{I}\vec{\tilde{R}}_J,
\label{act11}\end{equation}
where $\lambda^{I}{}_{K}=\lambda^{i}{}_{[k}\lambda^{j}{}_{\ell]}$, represents the second compound of $\lambda^{i}{}_{j}$.  With the same notation
(\ref{act10}) reads
\begin{equation}
\tilde{B}^I=\lambda^{I}{}_{J}B^J, \quad
\tilde{C}^{IJ}=\lambda^{I}{}_{K}\lambda^{J}{}_{L}C^{KL}.
\label{act12}\end{equation}

To obtain the invariants and covariants of the group action $\SO(4) \circlearrowright \mathcal{K}^2(\mathbb{S}^3)$, we apply the theory developed in \cite{JTH07}.  In particular, by taking contractions of products of the general Killing tensor, the Euclidean metric, and the coefficient tensor, we can obtain a complete set of $\SO(4)$-covariants and invariants.  Using matrices $\mathbi{K}$ and $\mathbi{C}$, where $(\mathbi{K})^i_{\ j} = K^i_{\ j}$, and $(\mathbi{C})^I_{\ J} = C^{ij}_{\ \ k\ell}$, as well as the trace operator ``Tr", we can express these covariants $\mathcal{C}_i$ and invariants $\mathcal{I}_i$ as follows
\[\mathcal{C}_1 = \mbox{Tr}(\mathbi{K}), \ \ \mathcal{C}_2 = \mbox{Tr}(\mathbi{K}^2), \ \ \mathcal{C}_3 = \mbox{Tr}(\mathbi{K}^3), \ \ \mathcal{C}_4 = g_{ij}x^ix^j,\]
\[\mathcal{I}_1 = \mbox{Tr}(\mathbi{C}), \ \ \mathcal{I}_2 = \mbox{Tr}(\mathbi{C}^2), \ \ \mathcal{I}_3 = \mbox{Tr}(\mathbi{C}^3), \ \ldots \ , \mathcal{I}_{14} = \mbox{Tr}(\mathbi{C}^{14}),\]
where, for example, $\mbox{Tr}(\mathbi{K}) = K^i_{\ i} $ and $\mbox{Tr}(\mathbi{C}) = C^{I}_{\ \ I}$.

As we stated in Section \ref{section1}, the action $\SO(4) \circlearrowright \mathcal{K}^2(\mathbb{S}^3)$ foliates the vector space $\mathcal{K}^2(\mathbb{S}^3)$ into the orbit space $\mathcal{K}^2(\mathbb{S}^3)/\SO(4)$.  Each orbit is represented by a canonical form, and the solution to the equivalence problem requires determining such canonical forms as well as a classification scheme for finding which orbit a given CKT belongs to.  As we mentioned in the introduction, the canonical forms problem has been solved.  Please refer to the appendix for a list of six canonical forms for the orbit space $\mathcal{K}^2(\mathbb{S}^3)/\mathrm{SO}(4)$.

Let us now develop a classification scheme for the orbit space $\mathcal{K}^2(\mathbb{S}^3)/\SO(4)$.  The set of invariants and covariants listed above could be used to try and classify the orbits of these CKTs.  This approach was successfully implemented in the solution to the equivalence problem of Killing tensors defined on $\mathbb{E}^3$ \cite{JTH07}, although the calculations were quite cumbersome.  Indeed, the difficulty in this approach lies in finding certain linear combinations of the above invariants and covariants which distinguish between the orbits.  A different approach which has proven to be more efficient and successful for the problems with small numbers of orbits is the \emph{method of web symmetries}. The central idea of this method is to use the symmetry properties of the associated orthogonal separable web of a canonical CKT to characterize its orbit.  In Section 2 we stated that the web symmetries of a Killing tensor $\mathbi{K}$ are generated by Killing vectors on the manifold.  Thus, to determine all of the symmetry generators of a given web we impose the following condition
\begin{equation}\label{symmetry} \mathcal{L}_\mathbi{V}\mathbi{K }= 0\end{equation}
on the Killing tensor $\mathbi{K}$ defining the web, using the general Killing vector $\mathbi{V}$ of our manifold.

In what follows we demonstrate the surprising result that the six CKTs of $\mathbb{S}^3$ can be classified based purely on the symmetry properties of their associated webs. To this end, we are interested in obtaining the symmetry properties of a web \emph{before} we impose the spherical constraint, which will yield additional web symmetries for the CKTs.  Visually, this corresponds to capturing all of the symmetry properties of a web before it is intersected with the surface of $\mathbb{S}^3$.

To achieve this, we impose condition (\ref{symmetry}) on each of the six CKTs using the general Killing vector of the ambient space $\mathbb{E}^4$.  This  will enable us to determine if a web is rotationally and/or translationally symmetry before it intersects the surface of $\mathbb{S}^3$.  Upon applying this method, we find that four of the six webs admit at least one rotational web symmetry. We can go even further by noting the number of rotational symmetries a CKT admits, which effectively divides the six canonical forms into three categories. Lastly, we find that two of the six webs admit translational symmetry, which provides the final distinguishing feature between each of the six webs. Please refer to Table \ref{S3classification} for a summary of these results.

\begin{remark}
In an application problem, it is possible that a given CKT $\mathbi{K}$ may have the Casimir tensor present. Specifically,
\[\mathbi{K} = \alpha \vec{\mathcal{C}} + \mathbi{K}_1, \]
where $\vec{\mathcal{C}}$ is the Casimir tensor (\ref{casimir1}) and $\mathbi{K}_1$ is a CKT.  If $\mathbi{K}_1$ is translationally symmetric, then the addition of the rotationally symmetric Casimir tensor destroys this translational symmetry.  Thus in order to determine all of the symmetries of $\mathbi{K}$ with or without the presence of the Casimir tensor, it is necessary to check the more general condition
\[\mathcal{L}_{\mathbi{V}}(\mathbi{K} + \alpha \vec{\mathcal{C}}) = 0,\]
for arbitrary $\alpha$.
\end{remark}

\begin{center}
\begin{table}\label{S3classification} \caption{Symmetry classification of the six webs of $\mathbb{S}^3$ under the action of $\SO(4)$} \begin{center}
\begin{tabular}{cccc}\hline \hline Category & Symmetry & Separable webs & Generators \\ \hline \hline I. &  2 rotations & cylindrical  & $\textbf{R}_{12}, \textbf{R}_{34}$ \\  II. & 1 translation \& 1 rotation & spherical & $\textbf{R}_{12}, \textbf{X}_{4}$ \\ III. & 1 translation & spheroelliptic & $\textbf{X}_4$ \\  IV. & 1 rotation & elliptic-cylindrical I  & $\textbf{R}_{12}$\\  & & elliptic-cylindrical II & \\  V.  & none & ellipsoidal  & \\ \hline \end{tabular} \end{center}
\end{table}
\end{center}

It is necessary to prove that the aforementioned symmetry properties of a Killing tensor are invariant under the action of $\SO(4)$.  To do so, it suffices to solve the equivalence problem of Killing vectors of $\mathcal{K}^1(\mathbb{E}^4)$ under the action of the group $\SO(4)$.  To begin, we note that the general Killing vector of $\mathbb{E}^4$ is given by
\[\mathbi{K} = A^i\mathbi{X}_i + B^{I}\mathbi{R}_I,\]
where $A^i$ and $B^I$ denote the Killing vector parameters, and $\mathbi{X}_i$ and $\mathbi{R}_I$ are the Killing vector fields defined previously.  The action of $\SO(4)$ on $\mathbi{K}$ induces the following transformations
 \[\tilde{A}^i = \Lambda^i_{\ j}A^j, \ \ \tilde{B}^I = \Lambda^I_{\ K} B^K,\]
 on the Killing vector parameters.  Therefore, it follows that
\[\mathcal{I}_1 = B^IB_I,\ \  \ \ \mathcal{I}_2 = A^iA_i\]
are invariants in the orbit space $\mathcal{K}^1(\mathbb{S}^3)/\SO(4)$. Using either of these two invariants it is possible to distinguish between two different types of symmetry generators.  Please refer to Table \ref{KVclassification} for a summary of these results. We can conclude that the translational and rotational web symmetries as defined by the Killing vectors of $\mathcal{K}^1(\mathbb{E}^4)$ are inequivalent under the action of $\SO(4)$.

\begin{center}
\begin{table} \label{KVclassification} \caption{Invariant classification of Killing vectors on $\mathbb{E}^4$ under the action of $\SO(4)$}\begin{center}
\begin{tabular}{ccc}  \hline \hline Category & Canonical form & $\mathcal{I}_1$  \\ \hline \hline  I & $\mathbi{R}_{12} $& $\neq 0$  \\ II & $\mathbi{X}_1$ & 0   \\ \hline
\end{tabular}
\end{center}
\end{table}
\end{center}

In addition to a classification scheme, a solution to the equivalence problem also requires a method for determining the \emph{moving frames map} which identifies the group action required to return a given CKT to the canonical form of its orbit.  On the two-dimensional manifolds $\mathbb{E}^2, \mathbb{M}^2$ and $\mathbb{S}^2$, algebraic formulas have been derived \cite{MST02, MST04, H08} for determining the moving frame map of a given CKT.  On $\mathbb{E}^3$ and $\mathbb{M}^3$, a combination of web symmetry and eigenvalues and eigenvectors of the parameter matrices has be used to determine such a map \cite{HMS05,HMS09}.  In our case, however, the situation is complicated by the fact that our coefficient tensor, $C^{ijk\ell}$, has order six when regarded as a matrix.  As such, we will need to devise a different strategy for determining the moving frame map of a CKT on $\mathbb{S}^3$.

It has been noted in Section 3  from (\ref{symrels}) that the coefficient tensor $C^{ijk\ell}$ has the same symmetries as the curvature tensor, and thus can be called an \emph{algebraic curvature tensor}.  In light of this property, let us lower the last three indices of $C^{ijk\ell}$ and contract on the first and third indices
\[Ric = C^{i}_{\ ji\ell} = R_{j\ell}\]
to obtain an \emph{algebraic Ricci tensor}.  The coefficient tensor for each of the six canonical forms listed in the Appendix can be contracted to define a \emph{canonical Ricci tensor} in each case.  The following proposition demonstrates that the Ricci tensor can be used to define the moving frame map for a given CKT.

\begin{proposition}\label{RicciProp}
A Killing tensor (\ref{gK}) is in canonical form if and only if its Ricci tensor is in canonical form.
\end{proposition}

\begin{proof}
Since the canonical form of the Ricci tensor is defined by the canonical form of $\mathbi{K}$, the first direction is  trivial. For the other direction, we prove by contradiction.  Suppose the Ricci tensor of a Killing tensor $\mathbi{K}$ is in canonical form, but $\mathbi{K}$ is not. Since $SO(4)$ acts transitively on the orbits of $\mathcal{K}^2(\mathbb{S}^3)/SO(4)$, we can find a group action $\Lambda \in SO(4)$ which sends $\mathbi{K}$ to its canonical form $\tilde{\mathbi{K}}$.  In particular, the components of $\mathbi{K}$ transform according to (\ref{act9})
which induces the following transformation
\[\tilde{C}^{i}_{\ jk\ell} = \Lambda^i_{\ m} \Lambda^n_{\ j} \Lambda^p_{\ k} \Lambda^{q}_{\ \ell} C^{m}_{ \ \ npq}\]
on the coefficient tensor $C^i_{\ jk\ell}$. At the same time, this action on $C$ induces the following transformation
\[ \tilde{R}_{j\ell} = \Lambda^m_{\ j}\Lambda^n_{\ \ell}R_{mn}\]
on its Ricci tensor $\mathcal{R}$. Since the Ricci tensor of a canonical Killing tensor is necessarily canonical, we must have $\tilde{\mathcal{R}} = \mathcal{R}$.  This is a contradiction.
\end{proof}

According to Proposition \ref{RicciProp}, the moving frame map of a CKT can be constructed by determining the moving frame map of the corresponding Ricci tensor.  Note that each canonical Ricci tensor can be represented by a diagonal matrix of order four.  Therefore, the determination of the moving frame map for the Ricci tensor is an eigenvalue-eigenvector problem for matrices of order four.  Before we illustrate this technique with the application in the next section, we summarize our results in the following algorithm.
\begin{enumerate}
\item Begin by substituting the potential into the compatibility condition (\ref{compatibility}) to determine the most general Killing tensor compatible with the potential.  Using this Killing tensor, determine the subspace of CKTs.
 \item Next, we classify a CKT $\vec{K}$ by determining whether it admits any symmetry. Namely, impose the constraint
 \[\mathcal{L}_{\mathbi{V}} (\mathbi{K}+ \alpha \mathbi{C}) = 0,\]
 where $\alpha$ is an arbitrary parameter, $\vec{\mathcal{C}}$ is the Casimir tensor, and $\vec{V}$ is the general Killing vector of $\mathbb{E}^4$,
 \[\mathbi{K} = A^i\mathbi{X}_i + B^{I}\mathbi{R}_{I}.\]
 If $\vec{K}$ does admit symmetry, determine which type and the number of generators for each type. Consult Table \ref{S3classification} to classify the CKT.
 \item To determine the moving frame map for $\mathbi{K}$, find the Ricci tensor $\mathcal{R}$ of the coefficient tensor.
 Diagonalize $\mathcal{R}$ by solving the corresponding eigenvalue-eigenvector problem.  The matrix $\Lambda$, which diagonalizes $\mathcal{R}$ defines the moving frame map.
 \item Finally, define the orthogonally separable set of coordinates corresponding to $\mathbi{K}$ by substituting $\Lambda$ found in the previous step into the equation
 \[x^i = \Lambda^i_{\ j} T^j(u^k),\]
 where $x^i = T^j(u^k)$ denote the canonical orthogonally separable coordinates corresponding to $\mathbi{K}$.
 \end{enumerate}
\section{Application}

Consider the following natural Hamiltonian
\[H = \mathcal{C}^{ij}p_ip_j + \frac{1}{(x-y)^2},\]
defined on $\mathbb{S}^3$, where $\mathcal{C}^{ij}$ denotes the Casimir tensor and $x,y,z,w$ are cartesian coordinates of the ambient space $\mathbb{E}^4$.  Using this Hamiltonian, we will now demonstrate how to apply the theory outlined in this paper.

First, we impose the compatibility condition (\ref{compatibility}) to obtain a family of Killing tensors which are compatible with the potential.  Of this family, the following restrictions on the parameters yields a subfamily of Killing tensors which satisfies the Haantjes condition (\ref{Haantjes1}) and generally admits 3 distinct eigenvalues:
\[\begin{array}{l} C_{1212} = C_{3434}, \ C_{1313} = C_{2323}, \ C_{1414} = C_{2424},\\  C_{1323} = C_{1313} - C_{1212}, \ C_{1424} = C_{1212} - C_{1414} \end{array}\]
Therefore we conclude that $K$ must characterize at least one of the six orthogonal separable webs of $\mathbb{S}^3$.  After a direct calculation, we  find that $K$ admits the following family of rotational Killing vectors
\[\mathbf{V} = (c_3z  + c_6w)\frac{\partial}{\partial x} + (c_6w - c_3z)\frac{\partial }{\partial y} + (c_3y - c_3x)\frac{\partial}{\partial z} - (c_6y+c_6x)\frac{\partial}{\partial w},\]
for arbitrary constants $c_3$ and $c_6$.  Using the classification scheme outlined in Table \ref{S3classification} we conclude that $K$ characterizes a non-canonical cylindrical web.

In order to determine the orthogonally separable coordinates for this Killing tensor, we need to determine the transformation which maps $K$ to its canonical form.  As discussed in Section \ref{section4}, such a map can be constructed by diagonalizing the Ricci tensor of the coefficient tensor.  Contracting indices, we obtain the following non-canonical Ricci tensor for this family of characteristic Killing tensors
\[\mathcal{R}_{jk} = \left(\begin{array}{cccc} C_{1212} + C_{1313} + C_{1414} & C_{1414} - C_{1313} & 0 & 0 \\  C_{1414} - C_{1313}  & C_{1212} + C_{1313} + C_{1414} & 0 & 0 \\ 0 & 0 & 2C_{1313} + C_{1212} & 0 \\ 0 & 0 & 0 & 2C_{1414} + C_{1212}  \end{array} \right).\]
After calculating the eigenvalues and corresponding eigenvectors of $\mathcal{R}_{jk}$ and applying the Gram-Schmidt orthonormalization procedure, we obtain an orthogonal matrix
\[\Lambda^i_{\ j} = \left(\begin{array}{cccc} \frac{-\sqrt{2}}{2} & 0 & 0 & \frac{\sqrt{2}}{2}\\ \frac{\sqrt{2}}{2} & 0 & 0&  \frac{\sqrt{2}}{2} \\ 0 & -1 & 0 & 0\\ 0 & 0 & 1 & 0 \end{array}\right)\]
which brings the Ricci tensor into canonical form
\[\tilde{\mathcal{R}}_{jk} = \mbox{diag}(2C_{1313} + C_{1212}, \ 2C_{1313} + C_{1212}, \ 2C_{1414} + C_{1212}, \ 2C_{1414} + C_{1212}).\]
Therefore, we conclude that
\[\begin{array}{lll} x &=& -\frac{\sqrt{2}}{2}\cos t\cos u + \frac{\sqrt{2}}{2}\sin t \sin v  \\ y &=& \frac{\sqrt{2}}{2}\cos t\cos u + \frac{\sqrt{2}}{2}\sin t \sin v \\ z &=& -\cos t \sin u \\ w &=& \sin t \cos v.\end{array} \]
is a system of orthogonally separable coordinates for this Hamiltonian.

\section{Conclusion}

The results presented in this paper conclude an important project  within the framework  of a more  general program of  the development of  the Hamilton-Jacobi theory of orthogonal separation of variables for natural Hamiltonians defined in spaces of constant curvature (see Table \ref{table:conclusions}). Having solved the equivalence problem, thus extending the classical result by Olevsky,  we have developed a general algorithm for solving the natural Hamiltonias defined in three-dimensional sphere  via orthogonal separation of variables.  In addition we give a simple and
concise proof of the fact that the validity of the first and second TSN conditions imply the validity of the third.
We have  also derived a set of  analogous algebraic conditions following from  the vanishing Haanjes tensor which can be used to study and characterize algebraic and geometric properties of  Killing two-tensors defined in spaces of constant, non-flat curvature. These conditions provide an alternative characterization of CKTs to the one derived in \cite{S10} and this paper based on the TSN criterion.

The results presented here lay the groundwork for a project that concerns orthogonal separation of variables afforded by characteristic Killing two-tensors defined in three-dimensional hyperbolic space which is the subject of a forthcoming paper.
\begin{table}[t]\begin{center}\renewcommand{\arraystretch}{1.2}
  \caption{\label{table:conclusions}Solutions to the canonical forms and
    equivalence problems.}
  \begin{tabular}{lll} \hline\hline
    \rule[-1.5mm]{0mm}{5mm} & Canonical forms problem & Equivalence problem \\ \hline
    Euclidean space $\Eset^{3}$ & $\begin{array}{l} \text{Eisenhart, 1934} \\
    \text{Boyer {\em et al\/}, 1976}  \end{array}$ & $\begin{array}{l} \text{Horwood {\em et al\/}, 2005}
    \\ \text{Horwood, 2007} \end{array}$ \\ \hline
    Minkowski space
    $\Mset^{3}$   & $\begin{array}{l} \text{Horwood and}
    \\ \text{McLenaghan, 2007} \end{array}$ &
    \text{Horwood {\em et al}, 2009} \\ \hline
    Sphere
    $\Sset^{3}$   &  \text{Olevsky, 1950}
   &
    \text{Cochran {\em et al}, 2010} \\

 \hline\hline
  \end{tabular}
\renewcommand{\arraystretch}{1.0}\end{center}\end{table}

\vspace{1in}

\noindent \Large \textbf{Acknowledgements} \\ 
\normalsize

\noindent  The authors acknowledge financial support from National Sciences Engineering Council of Canada (NSERC) in the form of Discovery Grants (RGM, RGS) and Postgraduate Scholarship (CMC) as well as the Izaak Walton Killam Memorial Scholarship (CMC).

\normalsize 


\section*{Appendix}
The following is a list of canonical forms for the six orthogonally separable coordinate systems of $\mathbb{S}^3$.  Webs II, III, IV and VI are defined by $\mathbi{K} = \alpha \mathbi{K}_1 + \beta \mathbi{K}_2 + \gamma \vec{\mathcal{C}}$, where $\mathbi{K}_1, \mathbi{K}_2$ come from Eisenhart's equations, $\vec{\mathcal{C}}$ is the Casimir tensor, and $\alpha, \beta, \gamma \in \mathbb{R}$.  Since it is possible to determine the presence of the Casimir tensor for webs I and V, and thus subtract it, these webs are defined by $\mathbi{K} = \alpha \mathbi{K}_1 + \beta \mathbi{K}_2$, where $\mathbi{K}_1, \mathbi{K}_2$ come from Eisenhart's equations and $\alpha, \beta \in \mathbb{R}$.
\subsection*{Rotational webs}
\begin{enumerate}[I.]
\item Spherical web
\subitem  $\mathbi{K} = c_1\mathbi{R}_{12}\odot\mathbi{R}_{12} + c_2(\mathbi{R}_{13}\odot \mathbi{R}_{13} + \mathbi{R}_{23}\odot \mathbi{R}_{23}) $
\subitem $R_{ij} = \mbox{diag}(c_1 + c_2, \ c_1 + c_2, \ 2c_2,0)$

\item Cylindrical web
\subitem $\mathbi{K} = c_1\mathbi{R}_{12}\odot \mathbi{R}_{12} + c_2(\mathbi{R}_{13}\odot\mathbi{R}_{13} + \mathbi{R}_{14}\odot\mathbi{R}_{14} + \mathbi{R}_{23}\odot\mathbi{R}_{23} + \mathbi{R}_{24}\odot\mathbi{R}_{24}) + $ \subitem $ \ \ \ \ \ \ \  c_3\mathbi{R}_{34}\odot \mathbi{R}_{34}$
\subitem $R_{ij} = \mbox{diag}(c_1 + 2c_2, \ c_1 + 2c_2, \ 2c_2 + c_3, \ 2c_2 + c_3)$
\item Elliptic-cylindrical web of type 1
\subitem $\mathbi{K} = c_1\mathbi{R}_{12}\odot \mathbi{R}_{12} + c_2(\mathbi{R}_{13}\odot \mathbi{R}_{13} + \mathbi{R}_{23}\odot \mathbi{R}_{23}) + c_3(\mathbi{R}_{14}\odot \mathbi{R}_{14} + \mathbi{R}_{24}\odot \mathbi{R}_{24})+$  \subitem $ \ \ \ \  \ \ \ c_4\mathbi{R}_{34}\odot \mathbi{R}_{34} $
\subitem Essential parameter: $k^2 = \displaystyle \frac{c_4-c_2}{c_4 - c_3}$
\subitem $R_{ij} = \mbox{diag}(c_1 + c_2 + c_3, \ c_1 + c_2 + c_3, \ 2c_2 + c_4, \ 2c_3 + c_4)$
\item Elliptic-cylindrical web of type 2
\subitem $\mathbi{K} = c_1\mathbi{R}_{12}\odot \mathbi{R}_{12} + c_2(\mathbi{R}_{13}\odot \mathbi{R}_{13} + \mathbi{R}_{23}\odot \mathbi{R}_{23}) + c_3(\mathbi{R}_{14}\odot \mathbi{R}_{14} + \mathbi{R}_{24}\odot \mathbi{R}_{24})+$  \subitem $ \ \ \ \  \ \ \ c_4\mathbi{R}_{34}\odot \mathbi{R}_{34} $
\subitem Essential parameter:  $k^2 = \displaystyle \frac{c_4 - c_3}{c_2 - c_3}$
\subitem $R_{ij} = \mbox{diag}(c_1 + c_2 + c_3, \ c_1 + c_2 + c_3, \ 2c_2 + c_4, \ 2c_3 + c_4)$
\subsection*{Translational web}
\item Spheroelliptic web
\subitem $\mathbi{K} = c_1\mathbi{R}_{12}\odot \mathbi{R}_{12} + c_2\mathbi{R}_{13}\odot \mathbi{R}_{13} + c_3\mathbi{R}_{23}\odot \mathbi{R}_{23}   $
\subitem Essential parameter:  $k'^2 = \displaystyle \frac{c_2 - c_3}{c_1 - c_3}$
\subitem $R_{ij} = \mbox{diag}(c_1 + c_2, \ c_1 + c_3, \ c_2 + c_3, 0)$
\subsection*{Asymmetric web}
\item Ellipsoidal web
\subitem $\mathbi{K} = c_1\mathbi{R}_{12}\odot \mathbi{R}_{12} + c_2\mathbi{R}_{13}\odot \mathbi{R}_{13} + c_3\mathbi{R}_{14}\odot \mathbi{R}_{14} +c_4\mathbi{R}_{23}\odot \mathbi{R}_{23}+ c_5\mathbi{R}_{24}\odot \mathbi{R}_{24} +$ \subitem \ \ \ \ \ \ \ $  c_6\mathbi{R}_{34}\odot \mathbi{R}_{34} $
\subitem where the $c_i$ satisfy the constraint
\subitem $\displaystyle (c_3 + c_4)(c_1c_6 - c_2c_5) + (c_2+ c_5)(c_3c_4 -c_1c_6) + (c_1 + c_6)(c_2c_5 - c_3c_4) = 0$
\subitem \subitem Essential parameters:
\subitem $a = \displaystyle \frac{c_1(c_2 - c_4) + c_6(c_2 - c_3) - c_2(c_3 + c_4) + 2c_3c_4}{c_1(c_2-c_4) + c_4(c_6-c_2) + c_5(c_4 - c_6)}, $
\subitem $b = \displaystyle \frac{c_2(c_1 - c_4) + c_1(c_5 - c_4) -c_3(c_1 + c_5) + 2c_3c_4}{c_1(c_2-c_4) + c_4(c_6-c_2) + c_5(c_4 - c_6)}$
\subitem $R_{ij} = \mbox{diag}(c_1 + c_2 + c_3, \ c_1 + c_4 + c_5, \ c_2 + c_4 + c_6, \ c_3 + c_5 + c_6)$
\end{enumerate}


\begin{thebibliography}{99}

\bibitem{AMS07} Adlam, C.~M., McLenaghan, R.~G., and Smirnov, R.~G., ``On
  geometric properties of joint invariants of Killing tensors'' (In) Eastwood,
  M. and Miller, W., Jr.\ (eds.) {\em Symmetries and Overdetermined Systems of
  Partial Differential Equations}, The IMA Volumes in Mathematics and its
  Applications {\bf 144}, 205--222 (Springer-Verlag, New York, 2008).

  \bibitem{Be97} Benenti, S., ``Intrinsic characterization of the variable
  separation in the Hamilton-Jacobi equation,'' J.\ Math.\ Phys.\ {\bf 38},
  6578--6602 (1997).

  \bibitem{Be1857} Bertrand J.~M., ``M\'{e}moire sur quelques-unes des forms les plus simples que
  puissent pr\'{e}senter les int\'{e}grales des \'{e}quations diff\'{e}rentielles du mouvement d'un point mat\'{e}riel,'' J.\ Math.\ Pures Appl.\ {\bf    2}, 113--140 (1857).

  \bibitem{BKM76} Boyer, C.~P., Kalnins, E.~G., and Miller, W., Jr., ``Symmetry
  and separation of variables for the Helmholtz and Laplace equations,''
  Nagoya Math.\ J.\ {\bf 60}, 35--80 (1976).

  \bibitem{EC37} Cartan, E., {\em La Th\'{e}orie des Groups Finis et Continus et la G\'{e}om\'{e}trie Diff\'{e}rentielle, Trait\'{e}es par
  la M\'{e}thode du Rep\`{e}re Mobile} (Gauthier-Villars, Paris, 1937).

  \bibitem{EC01} Cartan, E., {\em Riemannian Geometry in an Orthonormal Frame}
  (translation from Russian by V.~V.\ Goldberg) (World Scientific, Singapore,
  2001).

  \bibitem{delong} Delong, R.~P., Jr., {\em Killing tensors and the Hamiltion-Jacobi equation}, Ph.D. thesis,
  (University of Minnesota, 1982)

  \bibitem{E34} Eisenhart, L.~P., ``Separable systems of St\"{a}ckel,''
  Ann.\ Math.\ {\bf 35}, 284--305 (1934).

  \bibitem{FO98} Fels, M.~E., and Olver, P.~J., ``Moving coframes.\ I.\ A
  practical algorithm,'' Acta.\ Appl.\ Math.\ {\bf 51}, 161--213 (1998).

  \bibitem{FO99} Fels, M.~E., and Olver, P.~J., ``Moving coframes.\ II.\
  Regularization and theoretical foundations,'' Acta.\ Appl.\ Math.\ {\bf 55},
  127--208 (1999).

  \bibitem{G67} Greub, W.~H., {\em Linear Algebra} (Springer, Berlin-Heidelberg, 1967).

  \bibitem{PG74} Griffiths, P., ``On Cartan's method of Lie groups and moving
  frames as applied to uniqueness and existence questions in differential
  geometry,'' Duke Math.\ J.\ {\bf 41}, 775--814 (1974).

  \bibitem{H51} Haantjes, A., ``On $X_{n-1}$-forming sets of eigenvectors", Indag.\ Math.\ {\bf 17}, 158--162 (1955).

  \bibitem{HMS05} Horwood, J.~T., McLenaghan, R.~G., and Smirnov, R.~G.,
  ``Invariant classification of orthogonally separable Hamiltonian systems in
  Euclidean space,'' Commun.\ Math.\ Phys.\ {\bf 259}, 679--709 (2005).

  \bibitem{HM07} Horwood, J.~T., and McLenaghan, R.~G., ``Orthogonal separation
  of variables for the Hamilton-Jacobi and wave equations in three-dimensional
  Minkowski space,'' J.\ Math.\ Phys.\ {\bf 49}, 023501 (48~pages) (2008).

  \bibitem{HMS09} Horwood, J.~T., McLenaghan, R.~G., and Smirnov, R.~G., "Hamilton–Jacobi   theory in three-dimensional Minkowski space via Cartan   geometry," J.\ Math.\ Phys.\ {\bf 50}, 053507 (2009) (41 pages).

  \bibitem{JTH07} Horwood, J.~T., ``On the theory of algebraic invariants of
  vector spaces of Killing tensors,'' J.\ Geom.\ Phys.\ {\bf 58}, 487--501
  (2008).

  \bibitem{H08} Horwood, J.~T., {\em Invariant Theory of Killing Tensors} (PhD dissertation,
  University of Cambridge, 2008).



\bibitem{Ka86} Kalnins, E.~G., {\em Separation of Variables for Riemannian
  Spaces of Constant Curvature} (Longman Scientific \& Technical, New York,
  1986).

  \bibitem{KMW76} Kalnins, E.~G., Miller W., Jr. and Winternitz, P., ``The group O(4),
  Separation of Variables and the Hydrogen Atom,'' SIAM J.\ Appl. \ Math.\ {\bf 30}, 630--664 (1976).


  \bibitem{Li1846} Liouville, J., ``Sur quelques cas particuliers o\`{u} les \'{e}quations de mouvement d'un point mat\'{e}riel
  peuvents'int\'{e}grer,'' J.\ Math.\ Pures Appl.\ {\bf 11}, 345--378 (1846).

  \bibitem{MST02} McLenaghan, R.~G., Smirnov, R.~G., and The, D., ``Group
  invariant classification of separable Hamiltonian systems in the Euclidean
  plane and the $O(4)$-symmetric Yang-Mills theories of Yatsun,'' J.\ Math.\
  Phys.\ {\bf 43}, 1422--1440 (2002).
  
  \bibitem{MMS2004}  McLenaghan, R.~G., Milson, R., and Smirnov, R.~G., ``Killing tensors as
  irreducible representations of the general linear group,'' C. \ R. \ Acad. Sc. Paris, {\bf 339},
  621-624 (2004). 

  \bibitem{MST04} McLenaghan, R.~G., Smirnov, R.~G., and The, D., ``An extension of the classical theory of algebraic invariants to pseudo-Riemannian geometry and Hamiltonian mechanics,'' J. \ Math. \ Phys. \ {\bf 45}, 1079--1120 (2004).

  \bibitem{Mo1881}  Morera, G.,  ''Sulla separazione delle variabili nelle equazioni del moto
  di un punto materiale su una superficie,'' Atti Sci.\ di Torino {\bf 16}, 276--296 (1881).

  \bibitem{N1856} Neumann, C., ``De problemate quodam mechanico, quod ad primam integralium
  ultraellipticorum classem revocatur. Jour. Reine Angew. Math. 56 (1859), 46–63.

  \bibitem{N51} Nijenhuis, A., ``$X_{n-1}$-forming sets of eigenvectors,''
  Neder.\ Akad.\ Wetensch.\ Proc.\ {\bf 51A}, 200-212 (1951).

  \bibitem{O50} Olevsky, M.~N, ``Three orthogonal systems in spaces of constant
  curvature in which the equation $\Delta_{2} u + \lambda u = 0$ admits a
  complete separation of variables,'' Math.\ Sbornik {\bf 27}, 379--426 (1950).

  \bibitem{Olv99} Olver, P.~J., {\em Classical Theory of Invariants}
  (Cambridge University Press, Cambridge, 1999).

  \bibitem{RW05} Rauch-Wojciechowski, S., and Waksj\"{o}, C., ``What an effective
  criterion of separability says about the Calogero type systems,''
  J.\ Nonlin.\ Math.\ Phys.\ {\bf 12}, 535--547 (2005).

  \bibitem{S10} Sch\"{o}bel, K., ``Algebraic integrability conditions for Killing tensors on
  constant sectional curvature maniflolds,'' arXiv:1004.2872v1, (2010).

  \bibitem{Sc40} Schouten, J.~A., ``\"{U}ber Differentialkomitanten zweier
  kontravarianter Gr\"{o}ssen,'' Proc.\ Kon.\ Ned.\ Akad.\ Amsterdam
  {\bf 43}, 449--452 (1940).

  \bibitem{SY04} Smirnov, R.~G., and Yue, J., ``Covariants, joint invariants and
  the problem of equivalence in the invariant theory of Killing tensors defined
  in pseudo-Riemannian spaces of constant curvature,'' J.\ Math.\ Phys.\
  {\bf 45}, 4141--4163 (2004).

  \bibitem{St1891} St\"{a}ckel, P., {\em \"{U}ber die Integration der Hamilton-Jacobischer Differentialgleichung mitelst Separation der Variabeln} (Habilitationsschrift, Halle, 1891).




  \bibitem{WF65} Winternitz, P., and Fri\v{s}, I., ``Invariant expansions of
  relativistic amplitudes and subgroups of the proper Lorenz group,''
  Soviet J.\ Nuclear Phys.\ {\bf 1}, 636--643 (1965).


\end{thebibliography}
\end{document}